\documentclass[lettersize,journal]{IEEEtran}
\usepackage{amsmath,amsfonts}
\usepackage{amsthm}
\usepackage{algorithmic}
\usepackage[ruled,vlined]{algorithm2e}
\usepackage{array}

\usepackage[caption=false,font=footnotesize,labelfont=rm,textfont=rm]{subfig}
\usepackage{textcomp}
\usepackage{multirow}
\usepackage{threeparttable}
\usepackage{setspace}
\usepackage{bbding}
\usepackage{pifont}
\usepackage{stfloats}
\usepackage{subfloat}
\usepackage{subfig}
\usepackage{url}
\usepackage{verbatim}
\usepackage{makecell}
\usepackage{graphicx}
\usepackage{booktabs}
\usepackage{cite}
\usepackage{tikz}
\usepackage[hang,flushmargin]{footmisc}
\newtheorem{myThm}{Theorem}

\newcommand{\mathds}[1]{\text{\usefont{U}{dsrom}{m}{n}#1}}

\usepackage{xcolor} 
\usepackage{amsthm} 
\newtheoremstyle{example}
  {3pt} 
  {3pt} 
  {} 
  {} 
  {\bfseries} 
  {.} 
  {.5em} 
  {} 

\theoremstyle{example}
\newtheorem{example}{Example}

\usepackage{enumitem}
\setlist[itemize]{leftmargin=*}

\usepackage{xcolor} 
\usepackage{amsthm} 
\newtheoremstyle{Definition}
  {3pt} 
  {3pt} 
  {} 
  {} 
  {\bfseries} 
  {.} 
  {.5em} 
  {} 

\theoremstyle{Definition}
\newtheorem{myDef}{Definition}

\usepackage{xcolor} 
\usepackage{amsthm} 
\newtheoremstyle{Property}
  {3pt} 
  {3pt} 
  {} 
  {} 
  {\bfseries} 
  {.} 
  {.5em} 
  {} 

\theoremstyle{Property}
\newtheorem{myProp}{Property}

\usepackage{xcolor} 
\usepackage{amsthm} 
\newtheoremstyle{Proposition}
  {3pt} 
  {3pt} 
  {} 
  {} 
  {\bfseries} 
  {.} 
  {.5em} 
  {} 

\theoremstyle{Proposition}

\usepackage{xcolor} 
\usepackage{amsthm} 
\newtheoremstyle{Problem}
  {3pt} 
  {3pt} 
  {} 
  {} 
  {\bfseries} 
  {.} 
  {.5em} 
  {} 

\theoremstyle{Problem}
\newtheorem{Prob}{Problem}

\begin{document}

\title{Snoopy: Effective and Efficient Semantic Join Discovery via Proxy Columns}

\author{Yuxiang Guo,
        Yuren Mao,
        Zhonghao Hu,
        Lu~Chen,
        Yunjun~Gao, ~\IEEEmembership{Senior Member,~IEEE}
       
\thanks{Y. Guo,  L. Chen and Y. Gao are with the College of Computer Science, Zhejiang University, Hangzhou 310027, China (e-mail: guoyx@zju.edu.cn; luchen@zju.edu.cn; gaoyj@zju.edu.cn).}
\thanks{Y. Mao, Z. Hu are with the School of Software Technology, Zhejiang University, Ningbo 315048, China (e-mail: yuren.mao@zju.edu.cn; zhonghao.hu@zju.edu.cn).}
\thanks{*Corresponding author: Yunjun Gao (e-mail: gaoyj@zju.edu.cn)}
}



\maketitle

\begin{abstract}
Semantic join discovery, which aims to find columns in a table repository with high  semantic joinabilities to a given query column, plays an essential role in dataset discovery.
Existing methods can be divided into two categories: cell-level methods and column-level methods. However, both of them cannot simultaneously ensure proper effectiveness and efficiency. 
Cell-level methods, which compute the joinability by counting cell matches between columns, enjoy ideal effectiveness but suffer poor efficiency. In contrast, column-level methods, which determine joinability only by computing the similarity of column embeddings, enjoy proper efficiency but suffer poor effectiveness due to the issues occurring in their column embeddings: (i) semantics-joinability-gap, (ii) size limit, and (iii) permutation sensitivity. To address these issues, this paper proposes to compute column embeddings via proxy columns; furthermore, a novel column-level 
semantic join  discovery framework, \textsf{Snoopy}, is presented, leveraging proxy-column-based embeddings to bridge effectiveness and efficiency. Specifically, the proposed column embeddings are derived from the implicit column-to-proxy-column relationships, which are captured by the lightweight approximate-graph-matching-based column projection. To acquire good proxy columns for guiding the column projection, we introduce a rank-aware contrastive learning paradigm.
Extensive experiments on four real-world datasets demonstrate that \textsf{Snoopy} outperforms SOTA column-level methods by 16\% in Recall@25 and 10\% in NDCG@25, and achieves superior efficiency—being at least 5 orders of magnitude faster than cell-level solutions, and 3.5x faster than existing column-level methods.
\end{abstract}


\section{Introduction}
\label{sec:intro}

\IEEEPARstart{T}he drastic growth of open and shared datasets (e.g., government open data)
has brought unprecedented opportunities for data analysis. However, when confronted with massive datasets, users often struggle to find 
relevant ones for their specific requirements~\cite{WarpGate}.
Hence, the data management community has been developing dataset discovery systems to find related tables given a query table~\cite{TabelDiscovery}.

In this work, we focus on semantic join discovery.
Given a table repository $\mathcal{T}$, a query table $T_Q$ with the specified column $C_Q$,  semantic join discovery aims to find column $C$  from  $\mathcal{T}$ that exhibits a large number of semantically matched cells between column $C$ and $C_Q$. Thus, the table with the column $C$ can be joined with the query table $T_Q$ to enrich the features for data analysis/machine learing tasks.


\begin{figure}
  \centering
  \includegraphics[width=1\linewidth]{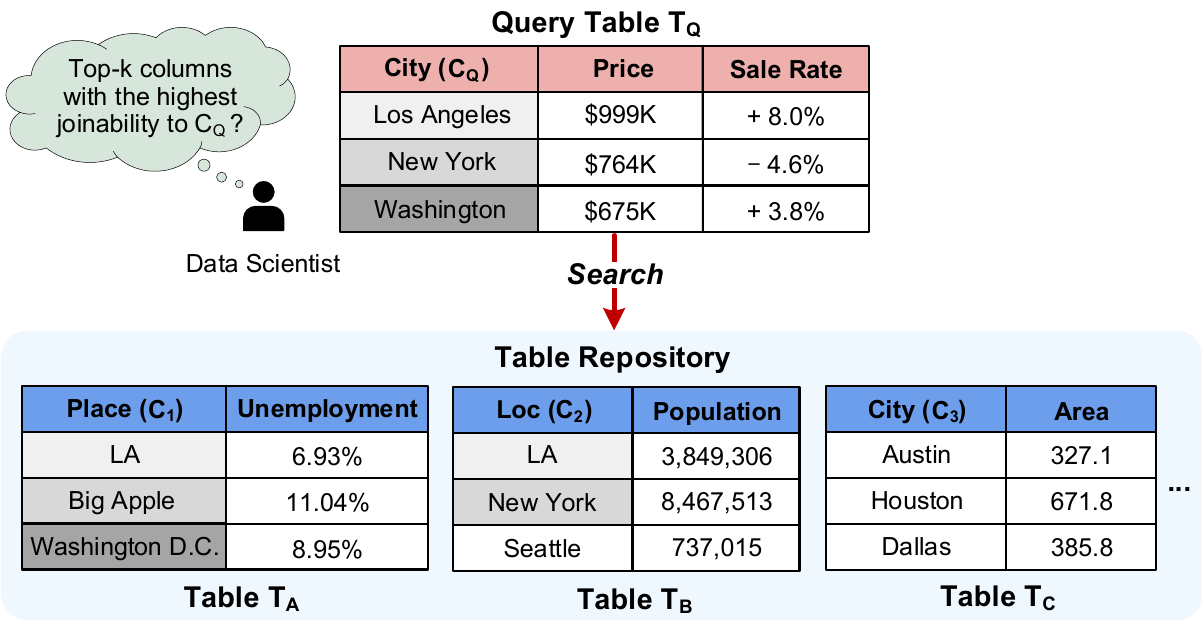} \vspace{-5mm}
   \caption{An example of semantic join discovery. The cells in columns $C_Q$, $C_1$ and $C_2$ with the same grayscale are matched.}
  \label{fig:exm1}
 \vspace{-6mm}
\end{figure}

\begin{example}\label{example-1}
As illustrated in Fig.~\ref{fig:exm1}, a data scientist wants to analyze the housing sales rate of the cities listed in Table $T_Q$. To enhance the understanding of these cities, she/he searches for tables with high joinability to Table $T_Q$ on the \texttt{City} column ($C_Q$). 
Since three (resp. two) cells in  $C_Q$ can find matched ones 
in $C_1$ (resp. $C_2$), the joinability of $C_Q$ and $C_1$ is $\frac{3}{3}$, while for $C_Q$ and $C_2$, it is $\frac{2}{3}$.
An ideal join discovery method returns $C_1$ as the top-1 column (assuming that only top-1 is requested). Then, the data scientist can employ any join algorithms~\cite{semanticjoin,autofuzzyjoin,Auto-Join} to merge the rows in Table $T_Q$ and $T_A$ based on the joinable columns, so that the \texttt{Unemplpyment} rate can be obtained for enhanced data analysis.
\end{example}



\definecolor{DeepGreen}{RGB}{0, 88, 36}



Most existing join discovery methods~\cite{JOSIE,LSH,DatasetDiscovery,CrossDataDis} limit the join type to the equi-join, overlooking numerous semantically matched cells (e.g., ``New York" semantically matches ``Big Apple"),  thus underestimating column joinabilities. To take semantics into consideration, 
cell-level semantic join discovery methods (e.g. PEXESO~\cite{Pexeso}) determine cell matching via cell embeddings, and compute the joinability based on the cell matching results. 
They are exact solutions under the proposed semantic joinability measure~\cite{Deepjoin,Pexeso}, but require online evaluation of all cell pairs between the query column and each repository column to assess column-to-column (c2c) joinabilities, as shown in Fig.~\ref{fig:exm2}(a).
Thus, they show ideal effectiveness but poor efficiency.
To enhance efficiency,  column-level  solutions~\cite{Deepjoin,WarpGate} propose to encode the entire column to a fixed-dimensional embedding via transformer-based pre-trained models (PTMs), as shown in Fig.~\ref{fig:exm2}(b),
and returns approximate results by performing similarity search on the column embeddings. While efficiency is improved, they suffer poor effectiveness due to the following reasons.
\begin{itemize}  
\item{} \textbf{Semantics-joinability-gap}: PTMs (e.g. SBERT adopted by DeepJoin~\cite{Deepjoin}) encode each column to an embedding independently, as shown in Fig.~\ref{fig:exm2}(b), without essential c2c interactions. This results in column embeddings inferior in capturing the cell information between columns, but focusing more on the semantic type of the input column~\cite{Watchog}. However, similarity in column semantic type does not necessarily imply high joinability. For instance, columns $C_Q$ and $C_3$ in Fig.~\ref{fig:exm1} share the same semantic type (\texttt{City}) but are not joinable, as no cells can be semantically matched.
\item{} \textbf{Size limit}: existing column-level methods concatenate cells within a column into a single sequence as input for PTMs. However, transformer-based PTMs typically impose an input size constraint, resulting in the truncation  of cells that exceed this limit, even if they contribute to the joinability. 
\item{} \textbf{Permutation sensitivity}: the embedding  obtained by these PTMs is sensitive to the permutation of tokens in the input sequence~\cite{tableembed}, however, the joinability of two columns is agnostic to the permutation of cells within each column.
\end{itemize}

\begin{figure}
  \centering
  \includegraphics[width=1\linewidth]{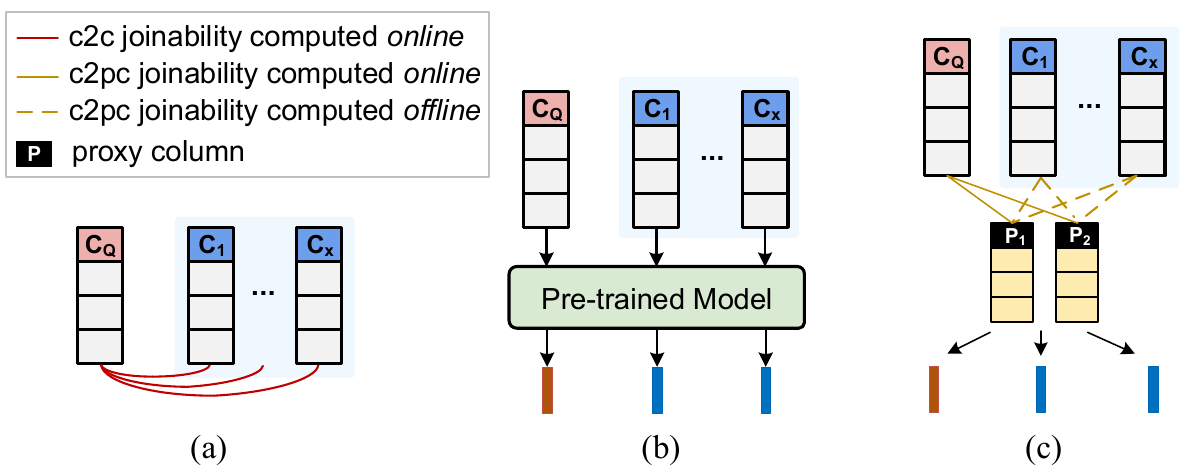} \vspace{-5mm}
   \caption{(a) The column-to-column  joinabilities that exact solutions should assess. (b) The existing column-level solutions encode columns independently without column-to-column interactions. (c)  Our proposed proxy columns capture column-to-proxy-column joinabilities and are used to derive column embeddings. For all repository columns $C_1, \dots, C_x$, c2pc values can be pre-computed offline. During the online stage,  only the computation $C_Q - P_i$ is needed to generate the embedding for the query column $C_Q$.}
   \vspace{-3mm}
   \label{fig:exm2}
\end{figure}

To address these issues, 
we introduce a novel concept termed \textbf{proxy columns}, which serve as representative columns in the column space.
Then, we define a  column projection function to capture the column-to-proxy-column (c2pc) joinabilities. 
The column embedding is derived as the concatenation of the column projection values guided by different proxy columns.
Introducing proxy-column-based column embedding offers several advantages:  (i) we can model c2pc joinabilities (cf. Fig.~\ref{fig:exm2}(c)) to represent the more general c2c  joinabilities, which can better estimate the joinabilities than the existing column-level methods; bridging the semantics-joinability-gap.  (ii) the derived column embeddings can overcome size limit and permutation sensitivity inherent in PTM-based encoders, as long as we guarantee the column projection function is size-unlimited and permutation-invariant; and  (iii) unlike c2c  joinabilities, which need online computations, c2pc joinabilities for all the repository columns can be pre-computed offline, making the online search highly efficient.
Under this main idea, two challenges arise.

\textbf{Challenge I: }\textit{How to design the column projection function to well model the c2pc joinabilities?}
The core of proxy-column-based column representation lies in how to design the column projection function.
A straightforward approach is to model c2pc joinabilities just like the c2c  joinabilities, i.e., defining the joinability score as the column projection value.
However, the joinability scores are typically low due to the strict join definition~\cite{Pexeso,Deepjoin}. Thus, the column embeddings which obtained by concatenating the projection values would be sparse and similar to each other, limiting their ability to capture different inter-column  joinabilities and impeding the effectiveness of subsequent similarity-based search.

\textbf{Challenge II: }\textit{How to obtain good proxy columns for similarity-based join search?} 
Since proxy columns directly affect the final column embeddings, it is crucial to find ``good" proxy columns that can map the joinable columns closely while effectively pushing non-joinable columns far apart in the column embedding space, facilitating the join discovery through similarity search of column embeddings.
In addition, good proxy columns should be aware of the   rank differences among joinable columns. For example, in Fig.~\ref{fig:exm1}, compared to the column $C_2$, good proxy columns should map the column $C_1$ closer to the query column $C_Q$.

To surmount these challenges, we present \textsf{Snoopy}, an effective and efficient column-level \underline{s}ema\underline{n}tic  j\underline{o}in  disc\underline{o}very framework via \underline{p}rox\underline{y} columns. Firstly, we formalize and analyze the desirable properties of column representations for column-level semantic join discovery. Then, we propose to construct a bipartite graph using the specific column and proxy column, and devise an approximate-graph-matching (AGM)-based column projection function to model the c2pc joinabilities. Since it is non-trivial to select a real textual column as a representative proxy column from the huge textual column space, we directly learn its matrix form, i.e., proxy column matrix. 
We regard proxy column matrices as learnable parameters, and present a rank-aware contrastive learning paradigm to learn the proxy column matrices which can derive column representations whose similarities imply joinabilities.  Additionally, we design two different training data generation methods to enable self-supervised training.

Our key contributions can be summarized as follows:
\begin{itemize}  
\vspace{-1mm}
\item{} \textit{Proxy-column-powered framework.} We propose a novel concept termed proxy columns. Based on this, we present a column-level semantic join discovery framework, \textsf{Snoopy}, which bridges effectiveness and efficiency.

\item{} \textit{Lightweight and effective column projection.}
We present a lightweight column projection function based on approximate graph matching, which effectively models the c2pc joinabilities and ensures the size-unlimited and permutation-invariant column representations.


\item \textit{Rank-aware learning paradigm.} We introduce a novel perspective that regards proxy column matrices as learnable parameters, and present a rank-aware contrastive learning paradigm to obtain good proxy columns.
Additionally, we design two training data generation strategies for self-supervised training.

\item{} \textit{Extensive experiments.} Extensive experiments on real-world table repositories demonstrate
that \textsf{Snoopy} achieves both high effectiveness and efficiency.

\end{itemize}

The rest of this paper is organized as follows.
Section~\ref{sec:pre} presents preliminaries. Section~\ref{sec:pivot} introduces proxy columns.  Section~\ref{sec:Snoopy} proposes the framework $\textsf{Snoopy}$. Section~\ref{sec:exp} reports experimental results. 
Section~\ref{sec:relatedwork} reviews the related work. Finally, Section~\ref{sec:conlusion} concludes the paper.

\section{Preliminaries}
\label{sec:pre}

In this section, we first present the the scope of semantic join discovery, and then provide the problem statement.

\subsection{The Scope of Semantic Join Discovery}
\label{subsec: scope}
Our goal is to design a framework that takes a query column as input, and searches for the top-$k$ semantically joinable columns from a large table repository. Then, it is easy to locate the tables with these joinable columns.
How to find the best way of joining column elements after discovery, as explored in studies~\cite{semanticjoin,autofuzzyjoin,Auto-Join} is out of our scope.
Note that, the join relationship is not a simple binary relationship (i.e., joinable or not joinable); rather, it often involves determining which column is more joinable to the query.
Since performing join is time-consuming, users and data scientists prefer to prioritize columns with higher joinability, rather than considering all columns labeled as joinable without any indication of their relative join probabilities.
Thus, it is beneficial for discovery methods to return the top-$k$ columns ranked by joinabilities.

\subsection{Problem Statement}
\label{subsec: prob}

Given a table repository $\mathcal{T}$, each table $T \in \mathcal{T}$ consists of several columns $\{C_1, C_2, \dots, C_{|T|}\}$, where $|T|$ denotes the number of columns in $T$.
In this paper, we focus on  textual columns, following the recent study~\cite{Deepjoin}.
For simplicity, we extract all the textual columns from $\mathcal{T}$ into a column repository $\mathcal{R} =\{C_1, C_2, \dots, C_{|\mathcal{R}|} \}$.  Each column $C_i$ consists of a set of cells $\{c_j^i\}$. Since metadata (e.g., column names) are typically missing or incomplete in real scenarios~\cite{NargesianZMPA19}, we rely solely on cell information.

\begin{myDef}
\label{def: colmat}
    \textnormal{\textbf{(Column Matrix)}.}  Given a column $C= \{c_1, c_2,  \cdots, c_n\}$, and a cell embedding function $h(\cdot)$ which transforms each $c_i \in C$ to an embedding $\mathbf{c}_i = h(c_i) \in \mathbb{R}^d$,  we extend the notation of cell embedding function $h(\cdot)$ to represent the column matrix $\mathbf{C} = h(C) = \{\mathbf{c}_1, \mathbf{c}_2, \cdots, \mathbf{c}_n\}$.
\end{myDef}

\begin{myDef}
\label{def:cellmatch}
    \textnormal{\textbf{(Cell Matching)}.} Given two cells $c_i$ and $c_j$, and the corresponding cell embeddings $\mathbf{c}_i$ and $\mathbf{c}_j$,  we denote $c_i$ matches $c_j$ as  
    $c_i$ $\cong$ $c_j$ if and only if $ d\big(\mathbf{c}_i, \mathbf{c}_j\big) \leq \tau $, where $d(\cdot)$ is a distance function, and $\tau$ is a pre-defined threshold.
\end{myDef}


 



 

\begin{myDef}
\label{def:js}
\textnormal{\textbf{(Semantic Joinability}).} Given a query column $C_Q$ from a table $T_Q$, and a column $C \in \mathcal{R}$, semantic joinability~\cite{Pexeso,Deepjoin} from $C_Q$ to $C$ is defined as follows:
\begin{equation}
\label{eq:js}
    J(C_Q, C)= \left|\left\{c_q \in C_Q \mid  \exists c \in C \text { s.t. } c_q \cong c\right\}\right|/\ |C_Q|  
\end{equation}
\end{myDef}

 
With the definition of semantic joinability measure, we can now formalize the problem we aim to solve.

\begin{Prob}
\label{prob}
    \textnormal{\textbf{(Top-$k$ Semantic Join Discovery)}\footnote{Equi-join discovery can be viewed as a specific case of semantic join discovery by setting the threshold $\tau$ in Definition~\ref{def:cellmatch} to 0.}.}  Given a query column $C_Q$ from a table $T_Q$, and a column repository $\mathcal{R}$, top-$k$  semantic join discovery aims to find a subset $\mathcal{S} \subset \mathcal{R}$ where $|\mathcal{S}|=k$  and $ \forall C \in \mathcal{S}$ and $C' \in \mathcal{R} - \mathcal{S}$, $J(C_Q, C) \geq J(C_Q, C')$.
\end{Prob}

Cell-level methods  (e.g. PEXESO~\cite{Pexeso}) design exact algorithms to Problem~\ref{prob} under the defined semantic joinability measure~\cite{Deepjoin}. These methods perform fine-grained cell-level computations to directly compute $J(C_Q, C_i)$ for the query column and each repository column $C_i \in \mathcal{R}$. Despite optimizations such as pruning and filtering, the worst-case time complexity remains linear with respect to the product of the query column size and the sum of repository column size.
In contrast, column-level methods perform a coarser computation and return approximate results. Specifically, they design a column embedding function $f: \mathcal{R} \rightarrow \mathbb{R}^l$ to obtain the query column embedding $f(C_Q)$ and repository column embeddings $f(C_i)$ for $C_i \in \mathcal{R}$, and perform similarity search to efficiently find the top-$k$ joinable columns $\hat{\mathcal{S}}$ ranked by $\operatorname{sim}\bigl(f(C_Q), f(C_i)\big)$, where $\operatorname{sim}(\cdot)$ is a  column-level similarity measurement. 
Hence, the primary challenge of column-level approaches is \textit{how to devise a column embedding function $f(\cdot)$ that makes the result $\hat{\mathcal{S}}$ as close as possible to $\mathcal{S}$.}



\label{subsec: pivot_tech}

\section{Proxy Columns}
\label{sec:pivot}
As mentioned in Section~\ref{sec:intro}, it is beneficial to model c2c  joinabilities. However, it is time-consuming to perform online computations to assess the relationship of the query column with each column in the repository. To tackle this, we introduce the concepts of proxy column and proxy column matrix.

\begin{myDef}
\textnormal{\textbf{(Proxy Column)}}. Given a column space $\mathbb{C}$ which is a collection of textual columns, a proxy column $P = \{p_1, p_2, \dots, p_m \} \in \mathbb{C}$ is a representative one, based on which, each column in the repository $\mathcal{R}$ can pre-compute the relationships with proxy columns. 
\end{myDef}

\begin{myDef}
\textnormal{\textbf{(Proxy Column Matrix)}}. Given a proxy column $P = \{p_1, p_2, \dots, p_m \}$ and a cell embedding function $h(\cdot)$, the proxy column matrix $\mathbf{P} = h (P) = \{\mathbf{p}_1, \mathbf{p}_2, \dots, \mathbf{p}_m \} \in \mathbb{R}^{m \times d}$, where $\mathbf{p}_i \in \mathbb{R}^d$. 
\end{myDef}

Note that the specific column repository $\mathcal{R}$ is just a subset of the column space $\mathbb{C}$.  Since \textsf{Snoopy} requires the column matrix and proxy column matrix as inputs (detailed later), we define the proxy-guided column projection given a proxy column matrix as follows.


\begin{myDef}
\textnormal{\textbf{(Column Projection)}}.
Given a proxy column matrix $\mathbf{P}$,  the proxy-guided column projection is a function $\pi_{\mathbf{P}} (\cdot)$, which projects a column matrix $\mathbf{C}$ to $\pi_{\mathbf{P}} (\mathbf{C}) \in \mathbb{R}$ indicating the value of a specific dimension in the column embedding space $\mathbb{R}^l$.
\end{myDef}

Based on the column projection, we can obtain the column embeddings. Specifically,   given a set of proxy column matrices  $\mathcal{P} = \{\mathbf{P}_1, \mathbf{P}_2, \dots, \mathbf{P}_l \} \in \mathbb{R}^{l\times m \times d}$, the column embedding of $\mathbf{C}$ is denoted as $
\phi(\mathbf{C})= \left[\pi_{\mathbf{P}_1} (\mathbf{C}), \pi_{\mathbf{P}_2}(\mathbf{C}), \ldots \pi_{\mathbf{P}_l} (\mathbf{C})\right] \in \mathbb{R}^l$.

\section{The \textsf{Snoopy} Framework}
\label{sec:Snoopy}
In this section, we first  
overview the \textsf{Snoopy} framework. Then, we design the column representation via proxy column matrices, and present a rank-aware contrastive learning paradigm to obtain good proxy column matrices. After that, we illustrate the index and online search process. Finally, we devise two training data generation strategies for self-supervised training.


\subsection{Overview}
\label{subsec:overview}

\textsf{Snoopy} framework is composed of two stages: offline and online, as illustrated in Fig.~\ref{fig:framework}.

\noindent\textbf{Offline stage.} Given a table repository $\mathcal{T}$, the textual columns are extracted to form the column repository $\mathcal{R}$. In the \ding{172} training phase, the column representation process transforms each column in $\mathcal{R}$ into a column embedding by concatenating the proxy-guided column projection values. Proxy column matrices are treated as learnable parameters and are updated via rank-aware contrastive learning.
\ding{173} After training, \textsf{Snoopy} uses the learned proxy column matrices to pre-compute all the column embeddings, and stores them in a Vector Storage. The indexes (e.g. HNSW~\cite{HNSW}) are constructed to accelerate the subsequent online search.

\noindent\textbf{Online stage.} Given a  query column $C_Q$ from table $T_Q$, \textsf{Snoopy} first computes the embedding of $C_Q$ using the previously learned proxy column  matrices. 
Then, it uses the embedding of $C_Q$ to search the top-$k$ similar embeddings from the Vector Storage. Finally, \textsf{Snoopy} returns the joinable columns according to the retrieved embeddings.

\begin{figure}
  \centering
  \includegraphics[width=1\linewidth]{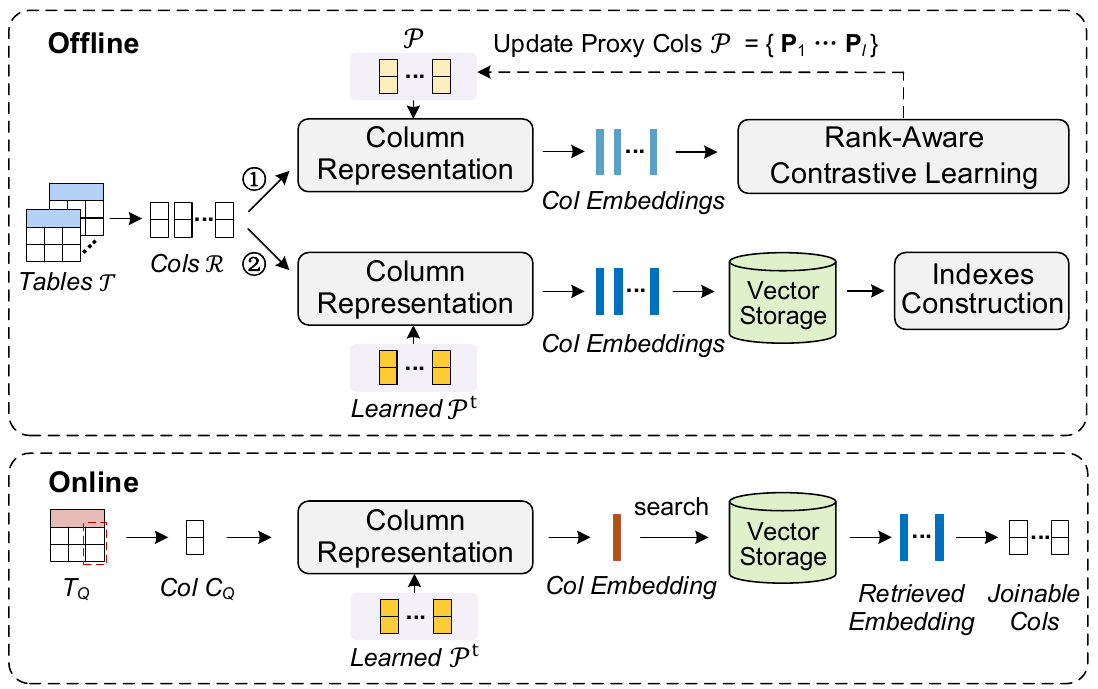} \vspace{-4mm}
  \caption{Overview of \textsf{Snoopy}. ``Column" is abbreviated as ``Col".}
  \label{fig:framework}
  \vspace{-4mm}
\end{figure}

\subsection{Column Representation}
The effectiveness of column-level methods highly depends on column representations. However, existing methods just adopt suboptimal PTMs as column encoders, and the desirable properties of column representations for column-level methods remain under-explored. Thus, we first formalize several desirable properties of column representations, and then propose the AGM-based column projection function to deduce the column embeddings.

\subsubsection{Desirable Properties}
We provide some key observations based on the definition of joinability and formalize the desirable properties of column representations for column-level semantic join discovery.

The first observation is that each cell in the column may contribute to the joinability score. Consider the example in Fig.~\ref{fig:exm1}, the joinability between  $C_Q$ and $C_1$ is $\frac{3}{3}$=1. If we neglect any cell in the  $C_1$, the join score declines. 
This observation implies that the column embedding function should consider all the cells within the column and not be
constrained by the column size, which is formalized as follows.

\begin{myProp}
\label{prop:1}
 \textnormal{\textbf{(Size-unlimited)}.}
     Given a column $C =\{ c_1, c_2,\\ \dots, c_n \} \in \mathcal{R}$, which is the input of the embedding function $f: \mathcal{R} \rightarrow \mathbb{R}^l$, the size $n$ of the input column $C$ is arbitrary.
\end{myProp}

The second observation is that the joinability between two columns is agnostic to the permutations of cells within each column. For instance, if we permute the column  $C_Q$  in Fig.~\ref{fig:exm1} from \{``Los Angeles", ``New York", ``Washington"\} to \{``Washington", ``Los Angeles", ``New York"\}, the joinability between $C_Q$  and $C_1$ is still $\frac{3}{3}$=1. This suggests that the column embedding should be permutation-invariant, which is formalized as follows.
\begin{myProp}
\label{prop:2}
\textnormal{\textbf{(Permutation-invariant)}.} 
Given a column $C = \{c_1, c_2, \dots, c_n \}$, and an arbitrary bijective function $\delta: \{1, 2, \dots, n\} \rightarrow \{1, 2, \dots, n\}$, a permutation of column $C$ is denoted as $\Tilde{C} = \{c_{\delta(1)}, c_{\delta(2)}, \dots, c_{\delta(n)} \}$. The column representation is expected to satisfy:
\begin{equation}
    f(C) \equiv f(\Tilde{C}) 
\end{equation}
\end{myProp}



Recall that the column-level methods return the top-$k$ results through the similarity comparison of column embeddings. 
Hence, an ideal column embedding function $f(\cdot)$ needs to preserve the ranking order of columns as determined by their joinabilities in Definition~\ref{def:js}.

\begin{myProp}
\label{prop:3}
\textnormal{\textbf{(Order-preserving)}.} 
Given a query column $C_Q$, two candidate columns $C_1$ and $C_2$, the joinability $J(\cdot)$ as defined in Definition~\ref{def:js}, and the column-level similarity measurement $\operatorname{sim}(\cdot)$, the ideal column representation satisfies:
\begin{equation}
\begin{gathered}
\operatorname{sim}\left(f\left(C_Q\right), f\left(C_1\right)\right) \geq \operatorname{sim}\left(f\left(C_Q\right), f\left(C_2\right)\right) \\
\Leftrightarrow J\left(C_Q, C_1\right) \geq J\left(C_Q, C_2\right)
\end{gathered}
\end{equation}

\end{myProp}

The property~\ref{prop:3} is the ultimate goal ensuring that column-level methods achieve the same effectiveness as exact solutions.
However, it is challenging to achieve this due to the inevitable information loss resulting from coarse computations at the column level.  
Thus, we design a training paradigm to learn good proxy column matrices capable of deriving column representations that approximate this goal (Section~\ref{subsubsec:PivotLearn}). Note that the defined joinability in Definition~\ref{def:js} relies on a cell embedding function $h(\cdot)$. If $h(\cdot)$ is inferior, then the ground truth in evaluation and the learned  column embedding function $f(\cdot)$ may be affected.
In this paper, we follow~\cite{Pexeso} to employ fastText~\cite{fasttext} as a cell embedding function, and show how the ground truth shifts when adopting different cell embedding functions (see Section~\ref{subsec:further_exp}).

\subsubsection{AGM-based Column Projection}
\label{subsec:ColRep}
In order to consider all the cells within the column $C$, 
\textsf{Snoopy} first transforms it into the column matrix $\mathbf{C} = h(C)$ by the cell embedding function $h(\cdot)$. 
Then, it computes the column
representation using $\mathbf{C}$ as the input. 
The cell embedding function $h(\cdot)$ is utilized to capture the cell semantics, so that the cells that are not exactly the same but semantically equivalence can be matched and contribute to the joinability.  
Several alternatives are available for the cell embedding function, including pre-trained language models~\cite{fasttext, sentencebert} and models tailored for entity matching~\cite{camper, ditto}.
Note that, designing a good cell embedding function is orthogonal to this work.


As mentioned in Section~\ref{sec:intro}, the core of proxy-column-based column representation lies in the design of column projection function  $\pi_{\mathbf{P}}(\cdot)$  that aims to well capture the c2pc relationships. Recall that, the semantic joinability in Equation (\ref{eq:js}) represents the c2c relationship captured by the cell-level methods, and it is naturally \textit{size-unlimited} and \textit{permutation-invariant} by definition. 
Thus, a straightforward way is to define $\pi_{\mathbf{P}}(\mathbf{C}) = J(\mathbf{C}, \mathbf{P})$. However, it results in significant information loss.  We begin our analysis by rewriting the joinability in Equation~(\ref{eq:js}) into the following equivalent  mathematical form:
\begin{equation}
\label{eq:js2}
    J(\mathbf{C}, \mathbf{P})= \sum_{i=1}^n \mathds{1}\left(\min _{\mathbf{p}_j \in \mathbf{P}}  d\left(\mathbf{c}_i, \mathbf{p}_j\right) \leq \tau\right) /|\mathbf{C}|
\end{equation}
\noindent
where $\mathds{1}(\mathbf{x})$ is a binary indicator function that returns 1 if the  predicate $\mathbf{x}$ is true  otherwise returns 0. Since the threshold $\tau$ is typically small~\cite{Pexeso} (otherwise, it would introduce many falsely matched cells), the predicate  $\left(\min _{\mathbf{p}_j \in \mathbf{P}}  d\left(\mathbf{c}_i, \mathbf{p}_j\right)\leq\tau\right)$  is typically NOT true.  Consequently, $\mathds{1}(\cdot)$ frequently returns a value of 0, resulting in $J(\mathbf{C}, \mathbf{P})$ being close to 0 for numerous different $\mathbf{C}$. Hence, the projection values tend to be small and similar to each other, resulting in a significant information loss.
\begin{example}
\label{example-4}
Consider the example in Fig.~\ref{fig:exm1}. According to the cells within the columns,  the column matrix $\mathbf{C}_Q$ is supposed to be similar to $\mathbf{C}_1$ while dissimilar to $\mathbf{C}_3$. Thus, we assume that  $\mathbf{C}_Q = [[-0.1, 0.2], [0.2, 0.3], [-0.2,$ $ 0.3]], \mathbf{C}_1 = [[-0.1, 0.25], [0.15, 0.35], $ $ [-0.15, 0.3]]$, and $\mathbf{C}_3 = [[-0.8, 0.2],$ $[0.3, 0.8], [0.5, 0.3]]$. We also assume that a proxy column matrix $\mathbf{P} = [[1, -1], [-1, -1]]$ is used, the Euclidean distance is adopted as $d(\cdot)$ and the threshold $\tau = 0.1$. If we set $\pi_{\mathbf{P}}(\mathbf{C}) = J(\mathbf{C}, \mathbf{P})$, we would get $\pi_{\mathbf{P}}(\mathbf{C}_Q) = \pi_{\mathbf{P}}(\mathbf{C}_1) = \pi_{\mathbf{P}}(\mathbf{C}_3) = 0$, despite the differences of the three columns.
\end{example}

This observation motivates us to explore a column projection mechanism to better capture the c2pc joinabilities. 
Since both the column $\mathbf{C}$ and proxy column $\mathbf{P}$ can be treated as order-insensitive cell sets, we resort to the maximum bipartite matching which has been extensively employed in set similarity measurement~\cite{SilkMoth,TokenJoin}. Given an input column $\mathbf{C} = \{\mathbf{c}_1, \mathbf{c}_2, \dots, \mathbf{c}_n\}$, a proxy column $\mathbf{P} = \{\mathbf{p}_1, \mathbf{p}_2, \dots, \mathbf{p}_m\}$, and a similarity measurement $\operatorname{sim}(\cdot)$, we construct a bipartite graph $\mathcal{G} = (\mathbf{C}, \mathbf{P}, E)$, where $\mathbf{C}$ and $\mathbf{P}$ are two disjoint vertex sets, $E = \{e_{ij}\}$ is an edge set where $e_{ij}$ connects vertices $\mathbf{c}_i$ and $\mathbf{p}_j$ and associated with a weight $\operatorname{sim}(\mathbf{c}_i, \mathbf{p}_j)$.
\begin{myDef}
    \textnormal{\textbf{(Maximum Weighted Bipartite Matching)}}. Given a bipartite graph $\mathcal{G} = (\mathbf{C}, \mathbf{P}, E)$, the maximum weighted bipartite matching aims to
    find a subset of edges   $M \subset E$ that maximizes the sum of edge weights while ensuring no two edges in $M$ share a common vertex. The matching value $\mathcal{M}(\mathbf{C}, \mathbf{P})$ is formulated as follows:
\begin{equation}
\label{eq:BM}
\begin{aligned}
&  {\mathcal{M}}(\mathbf{C}, \mathbf{P})=\max \sum_{i=1}^n \sum_{j=1}^m u_{i j} \operatorname{sim}\left(\mathbf{c}_{i},\mathbf{p}_{j}\right),  \text{subject to:} \\
& \textnormal{(i)} \; \forall  i \in\{1,2, \ldots, n\}, \sum_{j = 1}^m  u_{i j} \leq 1,  u_{i j} \in\{0,1\}\\
& \textnormal{(ii)} \; \forall j \in\{1,2, \ldots, m\}, \sum_{i = 1}^n  u_{i j} \leq 1,  u_{i j} \in\{0,1\}
\end{aligned}
\end{equation}
\end{myDef}
\noindent
where $u_{ij} = 1$ (or $0$) indicates the edge $e_{i j}$ is (or not) in $M$.

However, since the time complexity of weighted bipartite matching is 
$\mathcal{O}(\{\max(m,n)\}^3)$ using the classical Hungarian algorithm, it is rather inefficient to perform the computation.
To tackle this,
we remove the second constraint (ii) in Equation~(\ref{eq:BM}) to formulate an approximate graph matching (AGM) problem. In this way, two edges are allowed to have common vertices in set $\mathbf{P}$. Thus, the AGM can be solved by a lightweight greedy mechanism and we define the projection $\pi_{\mathbf{P}}(\mathbf{C})$ as the result  $\mathcal{M}'(\mathbf{C}, \mathbf{P})$ of AGM as follows:
\begin{equation}
\label{eq:ABM}
     \pi_{\mathbf{P}}(\mathbf{C}) = \mathcal{M}'(\mathbf{C}, \mathbf{P})=\sum_{i=1}^n  \max _{\mathbf{p}_j\in \mathbf{P}} \operatorname{sim}\left(\mathbf{c}_{i},\mathbf{p}_{j}\right)
\end{equation}
\noindent This approximation reduces the time complexity to $\mathcal{O}(mn)$, and the operations can be executed on a GPU, further enhancing efficiency. We use dot product as the measurement $\operatorname{sim}(\cdot)$.

Finally, given a proxy column set $\mathcal{P} = \{\mathbf{P}_1, \mathbf{P}_2, \dots, \mathbf{P}_l \}$, we can obtain the column embedding of column $C$:
\begin{equation}
\label{eq:pi}
    f(C)  =  \phi(\mathbf{C}) = \left[\pi_{\mathbf{P}_1} (\mathbf{C}), \pi_{\mathbf{P}_2}(\mathbf{C}), \ldots \pi_{\mathbf{P}_l} (\mathbf{C})\right] 
\end{equation}

\noindent where $\pi_{\mathbf{P}_i}(\cdot)$ is desinged as Equation (\ref{eq:ABM}).


\begin{myThm}
    The proposed column representation $f(C)$ is size-unlimited  and permutation-invariant.
\end{myThm}

\begin{proof}
   The  transformation $h(\cdot$) from the input column $C$ to column matrix $\mathbf{C}$ is size-unlimited, as all cells are processed independently. The function $\pi_{\mathbf{P}_i} (\cdot)$ in Equation (\ref{eq:ABM}) is also size-unlimited, as it will go through each cell embedding $\mathbf{c}_i$. Thus, $f(C)$ is size-unlimited.
    Permuting the order of cells in $C$ changes the order of $\{\mathbf{c}_i\}$ in column matrix $\mathbf{C}$. However, $\pi_{\mathbf{P}_i} (\cdot)$ is independent of this order,  making $\phi(\mathbf{C})$ permutation-invariant. Thus, $f(C)$ is permutation-invariant.
\end{proof}

\noindent \textbf{Discussion.}
We now analyze the correlation between the column projection in Equation (\ref{eq:ABM}) and the joinability definition in Equation (\ref{eq:js2}).
First, we substitute the distance function $d(\cdot)$ in Equation (\ref{eq:js2}) with a similarity function $\operatorname{sim}(\cdot)$ and correspondingly replace the distance threshold $\tau$ with a similarity threshold $\alpha$. Since the smaller distance means the higher similarity, we can rewrite the Equation (\ref{eq:js2}) as follows:
\begin{equation}
\label{eq:js3}
    J(\mathbf{C}, \mathbf{P})= \sum_{i=1}^n \mathds{1}\left(\max _{\mathbf{p}_j \in \mathbf{P}}  \operatorname{sim}\left(\mathbf{c}_i, \mathbf{p}_j\right)>\alpha\right) /|\mathbf{C}|
\end{equation}

It is observed that we can simplify 
the term $\sum_{i=1}^n \mathds{1}\left(\max _{\mathbf{p}_j \in \mathbf{P}}  \operatorname{sim}\left(\mathbf{c}_i, \mathbf{p}_j\right)>\alpha\right)$ in Equation (\ref{eq:js3}) by eliminating the inequality comparison condition $> \alpha$ and the step-like indicator function $\mathds{1}(\cdot)$, and obtain the Equation (\ref{eq:ABM}).
This indicates that our designed column projection is a
smoothed version of the primary component of Equation (\ref{eq:js3}).
It mitigates information loss from the step-like indicator function, preserving more c2pc relationship details.

\begin{example}
\label{example-6}
    Continuing with Example~\ref{example-4}, applying the AGM-based column projection yields  $\pi_{\mathbf{P}}(\mathbf{C}_1) = -0.3$, $\pi_{\mathbf{P}}(\mathbf{C}_2) = -0.55$, and $\pi_{\mathbf{P}}(\mathbf{C}_3) = 0.3$. Thus, the phenomenon of information loss is alleviated.
\end{example}

\subsection{Rank-Aware Contrastive Learning}
\label{subsubsec:PivotLearn}
To achieve Property~\ref{prop:3}, it is crucial to identify good proxy column matrices $\mathcal{P} \in \mathbb{R}^{l \times m \times d }$ that can map the joinable columns closer and push the non-joinable columns far apart in the column embedding space. Traditional pivot selection methods in metric spaces~\cite{ZhuCGJ22} can be extended to select proxy columns from the given column repository $\mathcal{R}$, and then derive the corresponding proxy column matrices. However, these methods are constrained in the subspace $\mathcal{R}$ of the column space  $\mathbb{C}$, and are not designed for order-preserving property, resulting in inferior accuracy (see ablation study in Section~\ref{sec:exp_ablation}). Also, it is non-trivial to directly select good proxy columns that can achieve order-preserving, as columns consist of discrete cell values. To this end, we present a novel perspective that regards the continuous proxy column matrices $\mathcal{P} \in \mathbb{R}^{l \times m \times d }$ as learnable parameters, and introduce the rank-aware contrastive learning paradigm 
to learn $\mathcal{P}$
guided by the order-preserving objective.


Order-preserving entails high similarity for joinable column embeddings and low similarity for non-joinable ones. We introduce the contrastive learning paradigm to achieve this objective. Specifically, given an anchor column $C$, a positive column $ C^{+}$ with high joinability $J(C, C^+)$,  a negative column set $\{C_i^-\}_{i=1}^u$ with low joinabilities $\{J(C, C_i^-)\}_{i=1}^u$, and the column embedding function $f(\cdot)$, we have InfoNCE loss~\cite{Moco} as follows:
\begin{equation}
\label{eq:cl}
\mathcal{L} =-\log \frac{e^{  f(C)^{\top} f(C^+)/ t}}{ e^{ f(C)^{\top} f({C}^+)/ t} + {\sum_i e^{  f({C})^{\top} f({C_i}^-)/ t}}}
\end{equation}
\noindent
where $t$ is a  temperature scaling parameter and we fix it to be 0.08 empirically. Note $ f(C)^{\top} f(C_i)$ denotes the  cosine similarity, as we normalize the column embedding to ensure that $\| f(\cdot) \| = 1$.  

\noindent \textbf{Rank-Aware Optimization.}
In traditional InfoNCE loss, each anchor column has just one positive column, and thus, it is difficult to distinguish the ranks of different positive columns~\cite{rank}. To incorporate the rank awareness of positive
columns, we introduce a positive ranking list $L_C = [C_1^+, C_2^+, \dots, C_s^+]$ for each anchor $C$. The positive columns within $L_C$ are sorted in descending order by joinabilities, i.e., $J(C, C_1^+) > J(C, C_2^+) > \dots > J(C, C_s^+)$. During training, each column $C_j^+ \in \mathcal{R}$ is regarded as a positive column or pseudo-negative column. Specifically, we recursively regard $C_j^+$ as a positive column and $\{C_{j+1}^+, C_{j+2}^+, \dots, C_s^+\}$ as pseudo-negative columns.
The pseudo-negative columns, along with the true negative columns $\{C_i^-\}_{i=1}^u$, are combined to form the negative columns for training, as shown in Fig.~\ref{fig:rank}.
Note that, the negative column set $\{C_i^-\}_{i=1}^u$ is generated from the negative queue $\mathcal{Q}$, which will be detailed  in Algorithm~\ref{algo:algo1} later.
As for the positive ranking list, we design two kinds of data generation strategies to automatically construct it (see Section~\ref{subsec:data_gen}).
We adopt the rank-aware contrastive loss~\cite{rank} $\mathcal{L}_{rank} = \sum_{j=1} ^ s \ell_j$, where $\ell_j$ is defined as follows:
\begin{equation}
\label{eq:rcl}
\ell_j =- \log \frac{e^{  f(C)^{\top}f(C_j^+)/ t}}{{ Z+\sum \limits_{r \geq (j+1)} e^{  f({C})^{\top} f({C_r}^+)/ t}}  }
\end{equation}

\noindent where $Z = { e^{ f(C)^{\top} f({C}{_j}^+)/ t} + {\sum_i e^{  f({C})^{\top} f({C_i}^-)/ t}}}$ in Equation~(\ref{eq:cl}). Equation~(\ref{eq:rcl}) takes the ranks into consideration. Specifically, minimizing $\ell_j$ requires
decreasing $f(C)^{\top} f(C_{j+1}^+)$; while minimizing $\ell_{j+1}$ requires increasing $f(C)^{\top} f(C_{j+1}^+)$. Thus, $\ell_j$ and $\ell_{j+1}$ will compete for the  $f(C)^{\top} f(C_{j+1}^+)$. Since $C_{j+1}^+$ would be regarded as negative columns $j$ times while $C_j^+$ would be regarded as negative columns only $(j-1)$ times,  it would guide a ranking of positive columns with $f(C)^{\top} f(C_j^+) > f(C)^{\top} f(C_{j+1}^+)$.

\begin{figure}
  \centering
  \includegraphics[width=1\linewidth]{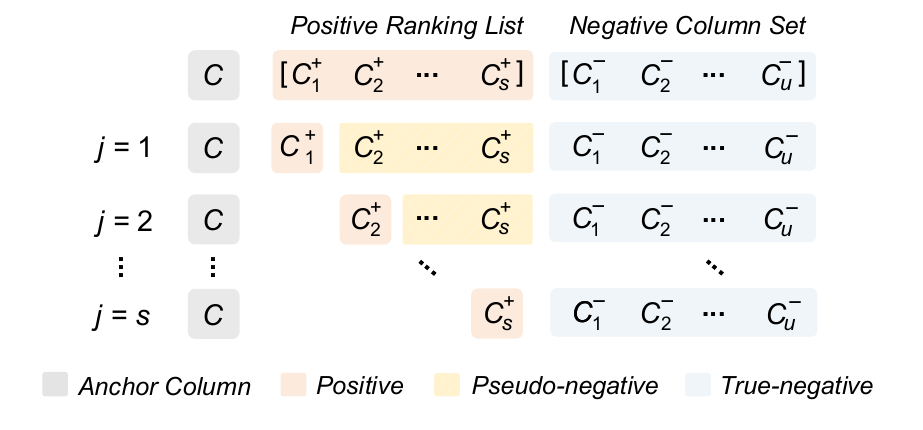} \vspace{-7mm}
  \caption{An illustration of rank-aware contrastive learning.}
  \label{fig:rank}
  \vspace{-1mm}
\end{figure}

\begin{algorithm}[!tb]

	\small
	\LinesNumbered
	\caption{\protect\mbox {\textsf{Rank-aware Contrastive Learning (RCL)}}}
 \label{algo:algo1}
        \KwIn{a column repository $\mathcal{R}$, the training epoch $\mathcal{E}$, the list $L$ of positive ranking lists for each column,  the length $s$ of each ranking list,
        the length $\beta$ of the queue $\mathcal{Q}$}
	\KwOut{the learned proxy column matrices  $\mathcal{P}^\text{t}$}
	Initialize target proxy column matrices  $\mathcal{P}^\text{t}$ and momentum proxy column matrices $\mathcal{P}^\text{m}$ in the same way\\
	\ForEach{ e $\in \left \{0, 1, ..., \mathcal{E}\right \}$}
    {     Get a $batch$ of columns from $\mathcal{R}$\\
	   $\mathcal{Q} \leftarrow$ Enqueue($\mathcal{Q}$, $batch$) \\
           \If{$len(\mathcal{Q}) >= \beta+1$}
           {
              CurrentBatch $B \leftarrow$ Dequeue($\mathcal{Q}$)\\
              \ForEach{ column $C$ in $B$  }
              {     Get the positive ranking list $L_C$ of $C$ from $L$\\
                   \ForEach{ j $\in \left \{0, 1, ..., s\right \}$  }
                   { 
                     $C_j^+ \leftarrow L_C\left[j\right]$\\
                     $\{C_r^+\} \leftarrow L_C\left[j+1:\right]$\\
                     $\{C_i^-\} \leftarrow$ All the columns in the current $\mathcal{Q}$\\
                     Compute $\ell_j$ using $C$, $C_j^+$,$\{C_r^+\}$, and $\{C_i^-\}$ \text{// Equation (\ref{eq:rcl}) }\\
                   }
                   Compute $\mathcal{L}_C$ using $\{\ell_j\}$
              }
              Compute $\mathcal{L}_B$ using $\{\mathcal{L}_C\}$\\
              Compute the gradients $\nabla_{{\mathcal{P}}^\text{t}}\mathcal{L}_B$\\
              Update $\mathcal{P}^\text{t}$ using stochastic gradient descent\\
              Update $\mathcal{P}^\text{m}$ with momentum\text{// Equation (~\ref{eq:mom}) }\\
              
           }
	
	}
	\Return{$\mathcal{P}^\text{t}$}
  
\end{algorithm}

In practice, considering the training efficiency, we freeze the cell embedding function $h(\cdot)$, and thus, the only learnable parameters are the proxy column  matrices $\mathcal{P} = \{\mathbf{P}_1, \mathbf{P}_2, \dots, \mathbf{P}_l \}$.
Since enlarging the number of negative samples typically brings performance improvement in contrastive learning~\cite{camper}, we maintain a  negative column queue $\mathcal{Q}$ with the pre-defined length $\beta$ to consider more negative columns.
We present the rank-aware contrastive learning process (\textsf{RCL}) in Algorithm~\ref{algo:algo1}.  At the beginning, \textsf{RCL} would not implement the gradient update until the  queue reaches the pre-defined length $\beta$ + 1 (line 5). Next, the oldest batch is dequeued and becomes the current batch (line 6).
For each column $C$ in the current batch, \textsf{RCL} gets the positive ranking list $L_C$ (line 8), and there are $s$ iterations to be performed to compute the loss $\mathcal{L}_C$. 
In each iteration, \textsf{RCL} first gets the positive column $C_j^+$, the pseudo-negative columns $\{C_r^+\}$, and  the true negative columns $\{C_i^-\}$ (line 10--12).
Then, \textsf{RCL} computes the loss $\ell_j$ using Equation (\ref{eq:rcl}). After $s$ iterations,  the loss $\mathcal{L}_C$ can be obtained (line 14). 
After getting losses of each column in the current batch, \textsf{RCL} computes $\mathcal{L}_B$ and the gradients of proxy column  matrices (line 15--16).
We adapt the momentum technique~\cite{Moco} to mitigate the obsolete column embeddings. Specifically, two sets of proxy column  matrices (i.e., the target proxy column  matrices ${\mathcal{P}}^\text{t}$ and the momentum proxy column matrices ${\mathcal{P}^\text{m}}$) are maintained. While ${\mathcal{P}}^\text{t}$ is instantly updated with the backpropagation, ${\mathcal{P}}^\text{m}$ is updated with momentum as follows:
\begin{equation}
\label{eq:mom}
    \mathbf{\mathcal{P}}^\text{m} \leftarrow \alpha \cdot \mathbf{\mathcal{P}}^\text{m}+(1-\alpha) \cdot \mathbf{\mathcal{P}}^\text{t}, \alpha \in[0,1)
\end{equation}

\noindent where  $\alpha$ is the momentum coefficient. Note that ${\mathcal{P}}^\text{t}$ and ${\mathcal{P}}^\text{m}$ are initialized in the same way before training (line 1).
The learned $\mathcal{P}^\text{t}$ is used to compute the column embeddings during offline and online processes.

\subsection{Index and Search}
\label{subsec:OnlineSearch}
After contrastive learning, \textsf{Snoopy} uses the learned $\mathcal{P}^\text{t}$ to pre-compute the embeddings of all the columns in the repository $\mathcal{R}$.
Then, \textsf{Snoopy} can construct the indexes for column embeddings using any prevalent indexing techniques, to boost the online approximate nearest neighbor (ANN) search. Since graph-based methods are a proven superior trade-off in terms of accuracy versus efficiency~\cite{WangXY021}, \textsf{Snoopy} adopts HNSW~\cite{HNSW} to construct indexes.

For online processing, when a query column $C_Q$ comes, \textsf{Snoopy} uses the learned $\mathcal{P}^\text{t}$ to get $f(C_Q)$. Compared with the existing PTM-based methods, \textsf{Snoopy}'s online encoding process is more efficient due to the proposed lightweight AGM-based column projection. Then, \textsf{Snoopy}  performs the ANN search using the indexes constructed offline and finds the top-$k$ similar column embeddings to $f(C_Q)$ from the Vector Storage. Finally, the corresponding top-$k$ joinable columns in the table repository are returned. We analyze the time complexity of  the online and offline stages in Appendix A.


\subsection{Training Data Generation}
\label{subsec:data_gen}
During contrastive learning, the negative column set is generated from the negative queue $\mathcal{Q}$, as most column pairs in the huge repository have low joinabilities.
As for the positive columns,
we design two kinds of data generation strategies.
Our objective is to synthesize positive columns with expected joinability scores. In this way, we can generate a ranked 
positive list $L_C$ according to given ranked  scores.

 
  
	
 


\vspace{1mm}
\noindent \textbf{Text-level Synthesis}.
Given a column $C$, we randomly divide it into two sub-columns $C_a$ and $C_b$. We use $C_a$ as the anchor column, and $C_b$ as the residual column. 
The reason why we maintain the residual column $C_b$ will be detailed later.
Now we illustrate how to synthesize a positive column $C_a^+$ for $C_a$ with a joinability approximated to a specified score $x\%$.

We randomly sample $x\%$ cells from the anchor $C_a$ to form a sampling column $S_a$. However, directly using $S_a$ as the positive column has two shortcomings: (i) the matched cells between $C_a$ and $S_a$ are exactly the same, without covering the semantically-equivalent cases, and (ii) each cell in $S_a$ can find matched cell in the anchor $C_a$, which is not realistic.
To tackle (i), we follow~\cite{Watchog,starmine} to apply straightforward augmentation operations to each cell $c \in S_a$ in text-level.  Note we should guarantee that the augmentation operator would not be too aggressive that $c$ and the augmented $c'$ are no longer matched under the distance threshold $\tau$ in Definition~\ref{def:cellmatch}. If that happens, we give up this augmentation and let $c'= c$. After that, we obtain an augmentation version $S'_a$ of $S_a$.
To tackle (ii), we concatenate $S'_a$ with the residual sub-column $C_b$ to obtain $S'_a || C_b$. Since we can assume that duplicates are few in the original column~\cite{autofuzzyjoin},  most cells in the sub-column $C_b$ do not have matches in the sub-column $C_a$, effectively simulating non-matched cells between the positive column and the anchor column.
Finally, we apply shuffling to $S'_a || C_b$ to obtain the required positive column $C_a^+$.

\vspace{1mm}
\noindent\textbf{Embedding-level Synthesis.}
An obstacle to text-level synthesis is that it is hard to determine the granularity of the applied augmentation operator, especially in an extremely small threshold $\tau$. To this end, we propose another embedding-level column synthesis strategy.

Given a column matrix $\mathbf{C}$, analogous to the text-level method, we first horizontally divide $\mathbf{C}$ into two sub-matrices $\mathbf{C}_a$ and $\mathbf{C}_b$. Then we randomly samples $x\%$ rows from $\mathbf{C}_a$, and denote it as $\mathbf{S}_a$. Note that each row of $\mathbf{S}_a$ is a cell embedding $\mathbf{c}$. Next, we  augment each cell embedding vector $\mathbf{c}$ by random perturbation. Specifically, we generate a random vector $\mathbf{r} $ following normal distribution $\mathcal{N}(0, \sigma)$, and gets the augmented $\mathbf{c}' = \mathbf{c} + \mathbf{r}$. Here, we also need a validation step to ensure that $d(\mathbf{c}, \mathbf{c}') \leq \tau$. But the advantage is that we can reduce $\sigma$ by multiplying a coefficient $\gamma \in (0,1)$ until the augmentation is proper to make $\mathbf{c}'$ matches $\mathbf{c}$ under the threshold $\tau$. Since the augmentation is operated in the continual embedding space, it is more flexible than the text-level method.

\section{Experiments}
\label{sec:exp}

In this section, we conduct extensive experiments to demonstrate the effectiveness and efficiency of our proposed \textsf{Snoopy}.

\subsection{Experiment setup}
\subsubsection{Datasets}

We use three real-world table repositories to evaluate the effectiveness of \textsf{Snoopy} and baseline methods. 
Wikitable is a dataset of relational tables from Wikipedia~\cite{Wikidataset}. Opendata is a data lake table repository from Canadian and UK open datasets~\cite{LSH,santos}. WDC Small is a sample dataset with long tables from the WDC Web Table Corpus\footnote{http://webdatacommons.org/webtables/2015/downloadInstructions.html}. For each dataset $\mathcal{T}$, we remove whitespace and only selected columns with a length greater than 5 and not numeric to form column repository $\mathcal{R}$. Since not all datasets contain metadata, for fairness, we only use the cells within each column. To generate queries and avoid data leaks, we randomly sample 50 columns for WikiTable and WDC Small, and 100 columns for Opendata from the original corpus except those in $\mathcal{T}$, following previous studies~\cite{Deepjoin, starmine}. 
For efficiency evaluation, we use a large dataset, WDC Large, which consists of 1 million columns extracted from 186,744 tables in the WDC Web Table Corpus.

\begin{table}[t] \small
\centering
\caption{Statistics of datasets. ``size" denotes \# of cells per column.}
\renewcommand{\arraystretch}{1.15} 
\vspace{-2mm}
\label{tab:dataset}
\setlength{\tabcolsep}{0.8mm}{
\begin{tabular}{l|ccccl} 
\specialrule{.12em}{.06em}{.06em}
Dataset   & $|\mathcal{T}|$  & $|\mathcal{R}|$ & Min. size & Max. size & Avg. size  \\ 
\hline
WikiTable &  32,614       &   228,299     &   5       &    999      &   29.41        \\
Opendata  &  2,310       &   13,918     &   5       &    1,417     &   151.74        \\
WDC Small &   17,763      &   97,703     &    27      &     810     &    102.44      \\ 
\hline
WDC Large &    186,744     &   1,000,000      &     1     &     1,974     &      115.90     \\

\specialrule{.12em}{.06em}{.06em}
\end{tabular}}
\vspace{-3mm}
\end{table}

\begin{table*}[t]
\footnotesize
\centering
\caption{Performance comparison of different methods on three datasets. The best results are in bold and the second best underlined. R@k and N@k refer to Recall@k and NDCG@k, respectively.}
\vspace{-2mm}
\label{tab:effectiveness}
\renewcommand{\arraystretch}{1.3} 
\setlength{\tabcolsep}{0.75mm}{
\begin{tabular}{c|ccc|ccc|ccc|ccc|ccc|ccc}  
\specialrule{.12em}{.06em}{.06em}
\multirow{2}{*}{\textbf{Methods}}  & \multicolumn{6}{c|}{\textbf{Wikitable}}         &\multicolumn{6}{c|}{\textbf{Opendata}}         & \multicolumn{6}{c}{\textbf{WDC Small}}           \\ 
\cline{2-19}
 & \textbf{R@5} & \textbf{R@15} & \textbf{R@25} & \textbf{N@5} & \textbf{N@15} & \textbf{N@25} & \textbf{R@5} & \textbf{R@15} & \textbf{R@25} & \textbf{N@5} & \textbf{N@15} & \textbf{N@25} & \textbf{R@5} & \textbf{R@15} & \textbf{R@25} & \textbf{N@5} & \textbf{N@15} & \textbf{N@25} \\ 
\hline
WarpGate      &  0.3720   &   0.5240   &  0.6072    &   0.8265  &  0.7806    &  0.7548   & 0.3440  &  0.6560  &   0.7416   &  0.9069   &  0.8824    &   0.8427   &  0.5240   &   0.6360   &   0.7040   &  0.9364   &    0.9307  &   0.9154    \\
BERT   &  0.3720   &   0.5067   &   0.5896   &   0.7851  &  0.7806    &   0.7548  & 0.3100    &   0.6000   &   0.7328   &  0.8780   &   0.8577   &  0.8397    &  0.5919   &    0.6107  &   0.6592   &  0.9065   &   0.8821   &   0.8709    \\
$\text{BERT}^* $  &  0.3920   &   0.5267   &   0.5840   &  0.8038   &   0.7912   &  0.7591  &  0.3500   &   0.6580   &   0.7688   &  0.8955   &  0.8853    &   0.8653    &  0.6080   &   0.6373   &    0.7168  & 0.9296    &   0.9236   &    0.9124   \\
SBERT   &  0.3480   &   0.4147   &   0.4800   &  0.7226   &  0.6705    &   0.6446 & 0.3120 & 0.5593  & 0.6728  & 0.8511  &   0.8176  &  0.7892    &   0.5560  &   0.5507   &   0.5816   & 0.9007    &  0.8547    &    0.8229   \\
$\text{SBERT}_\text{ck}$  &   0.3080 &	0.4133	& 0.4736   & 0.6796 &	0.6503	& 0.6316 &  0.3440	& 0.5927	& 0.7220 & 0.8620 &	0.8430 &	0.8265 &  0.5360	 & 0.5147 &	0.5456 & 0.8856 & 0.8320 & 0.7993  \\
$\text{Starmie}$       &   0.3760  &   0.4947   &  0.5496    &   0.7257  &   0.7070   &   0.6820  &    0.4240 &  0.6680    &  0.7936    &   0.9023  &   0.8872   &  0.8859    &  \underline{0.6120}   & 0.6267     &   0.7024   &   0.9159  &    0.9094  &    0.9007   \\
$\text{DeepJoin}$     &   0.3520  &  0.4827    &   0.5600   &   0.7347  &  0.7099    &    0.6984  &  \underline{0.4340}   &   0.6967  &    0.8188  &  0.9051   &   0.9017   &  0.9040     &  \textbf{0.6200}   &   0.6507   &   0.7240   &    0.9238 &   0.9227   &    0.9170   \\
$\text{DeepJoin}_\text{ck}$ & 0.3720 &	0.4827&	0.5384	&0.7312&	0.7039&	0.6790& \textbf{0.4400}&	0.6940&	0.8264 &\underline{0.9282}	&0.9095	&0.9117&0.6040&	0.6400&	0.7320&	0.9406	& 0.9273&	0.9242 \\
\hline
CellSamp&0.4360 & 0.6027 & 0.7264 & 0.8935	& 0.8960 & 0.8938 & 0.3620 & 0.6687& 0.8480 & 0.9157	& 0.9211 & 0.9281  & 0.3320 &	0.4813 &	0.5864 &	0.8314 &	0.8351	 & 0.8228 \\

\specialrule{.12em}{.06em}{.06em}
$\textsf{\textbf{Snoopy}}_\text{bs}$      &  \underline{0.4600}   &    \underline{0.6587}  &   \underline{0.7360}   &   \underline{0.9103}  &   \underline{0.9231}   &  \underline{0.8945}   &   0.3860    &  \underline{0.7113}    &   \underline{0.8552}   &   0.9259  &  \underline{0.9380}    &  \underline{0.9370}     &  0.5920  &   \underline{0.6813}   &   \underline{0.7960}   &  \underline{0.9503}   &   \underline{0.9500}   &    \underline{0.9472}   \\
\textsf{\textbf{Snoopy}}   & \textbf{0.5000}  &   \textbf{0.6600}   &  \textbf{0.7728}   &  \textbf{0.9187}   &   \textbf{0.9243}   &   \textbf{0.9122}  &  0.3920     &  \textbf{0.7180}    &  \textbf{0.8716}    &  \textbf{0.9300}   &   \textbf{0.9440}   &   \textbf{0.9458}   &   0.5760   &   \textbf{0.6827}   &  \textbf{0.8230}    &  \textbf{0.9613}   &  \textbf{0.9600}    &    \textbf{0.9632}   \\


\specialrule{.12em}{.06em}{.06em}
\end{tabular}}
\end{table*}

\subsubsection{Baselines} The following baselines are evaluated. 
\begin{itemize} 

\item{} \textbf{WarpGate}~\cite{WarpGate} is the latest system prototype for dataset discovery, which suggests using pre-trained table embedding models as column encoders. Since the embedding model~\cite{0002TGL21} used in the original paper is not available, we choose TURL~\cite{turl}, a well-adopted table embedding model.
\item{} \textbf{BERT}~\cite{bert} is a pre-trained language model. We use \textit{bert-base-uncased}\footnote{https://huggingface.co/bert-base-uncased} to get the embedding for each column, and use the default input length limit of 512.
\item{} $\textbf{BERT}^*$  adopts the contrastive loss to fine-tune BERT for the semantic join discovery task.
\item{} \textbf{SBERT}~\cite{sentencebert} is a specialized variant of BERT that is designed for sentence-level embeddings. We use MPNet~\cite{mpnet} as the backbone and set the input size limit as 512.
\item{}  $\textbf{SBERT}_\text{ck}$ is a variant of  SBERT , which divides the unique column values by chunks to ensure the 512-token limit, and averages chunk embeddings from SBERT to derive final column embeddings. 
\item{} \textbf{DeepJoin}~\cite{Deepjoin} is a state-of-the-art join discovery method, which fine-tunes SBERT to obtain column embeddings and samples the most frequent cells for each column to ensure the 512-token limit. We use the best-performance MPNet~\cite{mpnet} as the backbone, following~\cite{Deepjoin}.
\item{} $\textbf{DeepJoin}_\text{ck}$ is a chunking-based variant of DeepJoin. It employs the chunking strategy only during inference, as mini-batch training is not supported with chunking. 

\item{} $\textbf{Starmie}$~\cite{starmine} is a state-of-the-art dataset discovery method. It contextualizes column representations via fine-tuning RoBERTa to facilitate
unionable table search in data lakes. For fairness, in our evaluation, we employ the single-column encoder and fine-tune it for semantic join discovery.  

\item{} $\textbf{\textsf{Snoopy}}_\text{bs}$ is a base version of \textsf{Snoopy}, which adopts the traditional contrastive loss in Equation (\ref{eq:cl}) for training.
\item{} \textbf{CellSamp} is a cell-level method with sampling  by randomly selecting $n_s = 20$ cells per column to improve efficiency. It compute pairwise similarity scores of sampled cells, and the joinability score is determined as the average of the maximum similarity scores between each cell in the query column and all cells in the candidate column.
 
\end{itemize}

We implement baselines following their original settings and tune the parameters for the best-performing results. Since the PTM-based methods need textual sequences as input for fine-tuning, we use the proposed text-level synthesis strategy to construct the same training data for $\text{BERT}^*$, DeepJoin, Starmie, and $\textsf{Snoopy}_\text{bs}$. 
We did not observe accuracy improvements by fine-tuning TURL with our training data, as it requires some extra information such as captions and entities in the knowledge graph~\cite{turl}. Consequently, we directly utilize the pre-trained TURL to implement WarpGate.



\subsubsection{Metrics} Following Deepjoin~\cite{Deepjoin}, we adopt Recall@$k$ and NDCG@$k$ to evaluate effectiveness, where $k$ is set to 5, 15, and 25 by default. Recall@$k$ is defined as $\frac{|\hat{\mathcal{S}} \cap \mathcal{S}|}{k}$, where $\hat{\mathcal{S}}$ and  $\mathcal{S}$ denote the top-$k$ result obtained by the specific method and an exact solution, respectively.
NDCG@$k$ is defined as  $\frac{\text{DCG@}k}{\text{IDCG@}k}$,
where $\text{DCG@}k=\sum_{i=1}^k \frac{J(C_Q, \hat{C}_i)}{\operatorname{log}_2(i+1)}$ and $\text{IDCG@}k=\sum_{i=1}^k \frac{J(C_Q, {C}_i)}{\operatorname{log}_2(i+1)}$, and $\hat{C}_i$ and $C_i$ denote the columns ranked in the $i$-th position in the results obtained by the specific method and an exact solution, respectively. For efficiency, we evaluate the runtime of all the methods. The above metrics are averaged over all the queries.

\subsubsection{Implementation Details} We implement \textsf{Snoopy} in PyTorch. We set the batch size to 64, the length $\beta$ of the negative queue to 32, and the momentum coefficient $\alpha$  to 0.9999. We use Adam~\cite{adam} as the optimizer, and set the learning rate to 0.01. We set the number $l$ of proxy column matrix to 90, and the cardinality $m$ per proxy column matrix to 50 by default.  For $\textsf{RCL}$, we first generate
a list of sorted joinability scores, where each score is randomly generated from (0.6, 0.9), and then apply the embedding-level synthesis strategy to generate positive columns. The length of the positive ranking list is set to 3. For cell matching, we follow previous studies~\cite{Deepjoin,Pexeso} to use fastText~\cite{fasttext} as  cell embedding function, and normalize all the vectors to unit length. 
We use Euclidean distance as the distance function $d$, and the threshold $\tau$ of cell matching is set to 0.2 by default. All experiments were conducted on a computer with an Intel Core i9-10900K CPU, an NVIDIA GeForce RTX3090 GPU, and 128GB memory.  The programs were implemented in Python\footnote{
The source code and datasets are available at 
https://github.com/ZJU-DAILY/Snoopy.}.

\subsection{Effectiveness Evaluation}
\label{sec:exp_effectiveness}

Table~\ref{tab:effectiveness} presents the Recall@$k$ and NDCG@$k$ of $\textsf{Snoopy}$ and other baseline methods.

\noindent \textbf{$\textsf{Snoopy}_\text{bs}$ vs. competitors.}
The first observation is that the proposed $\textsf{Snoopy}_\text{bs}$ consistently outperforms other PTM-based methods across almost all evaluation metrics on the three datasets. 
Specifically, $\textsf{Snoopy}_\text{bs}$ demonstrates an average improvement of  12\% in Recall@25 and 8\% in NDCG@25, compared with the best column-level competitor. We contribute this improvement to the high-quality column embeddings derived by AGM-based column projection, which is capable of capturing the implicit relationships between column pairs. Note that, while $\textsf{Snoopy}_\text{bs}$ exhibits slightly lower Recall@5 compared to certain baselines on specific datasets, its NDCG@5 consistently outperforms them. The discrepancy stems from the fact that some joinable columns with the same joinability are ordered sequentially in the returned results. When $k$ is small, Recall@$k$ overlooks columns ranked beyond the $k$-th position, despite their equivalent joinability to those ranked at $k$. 
The second observation is that chunking does not always lead to performance improvements. Specifically, $\text{SBERT}_\text{ck}$ outperforms $\text{SBERT}$ on Opendata but performs worse on the other two datasets.  Similarly, $\text{DeepJoin}_\text{ck}$ surpasses $\text{DeepJoin}$ on Opendata and WDC Small but underperforms on WikiTable. This is because the chunking-and-averaging strategy inevitably leads to information loss due to the naive averaging operation, sometimes negating the benefits of incorporating additional cells.
The third observation is that the sampling-based cell-level method, CellSamp, performs well on the Wikitable and Opendata datasets due to its finer-grained computations compared to column-level methods.
However, it underperforms on the WDC Small dataset, where the shortest column contains 27 cells, exceeding CellSamp's sampling threshold of 20. While increasing the sampling threshold could incorporate more cells, this would substantially reduce the efficiency of the online search, as even with 20 cells, the computational cost remains orders of magnitude higher than column-level methods (see Table~\ref{tab:online_efficiency}).

\noindent \textbf{$\textsf{Snoopy}$ vs. competitors.}
With rank-aware optimization, $\textsf{Snoopy}$ outperforms the existing SOTA column-level methods by  16\%  in Recall@25 and 10\% in NDCG@25 on average. This is because the rank-aware optimization takes a list of joinable columns into consideration, which enables $\textsf{Snoopy}$ to better distinguish the ranks of different joinable columns.

To demonstrate the superiority of our column embeddings in bridging the semantics-joinability-gap, being size-unlimited, and permutation-invariant, we conduct in-depth analyses.
\noindent \textbf{Bridge the semantics-joinability-gap.}
First, we randomly select 1K columns from Opendata and use the trained Deepjoin and   $\textsf{Snoopy}$ to encode these columns into embeddings. We then visualize these embeddings using t-SNE, as shown in Fig.~\ref{fig:visual_embed}. The embedding distribution of Deepjoin is more uniform, indicating less evident similarity differences compared to $\textsf{Snoopy}$. This is because Deepjoin's embeddings focus more on the column semantic types rather than cell semantics, resulting in many columns having similar and indistinguishable embeddings.
We then compute the joinability of each query column with its top-25 joinable columns from Opendata, and the similarity between each query column's embedding and its top-25 joinable columns' embeddings obtained by Deepjoin and $\textsf{Snoopy}$.
Fig.~\ref{fig:visual_dist} depicts the distributions of joinability and column embedding similarity. It is observed that, the similarity distribution of Deepjoin's embeddings poorly fits the joinability distribution. In contrast, our $\textsf{Snoopy}$ demonstrates a closer alignment with the joinability distribution, effectively bridging the semantics-joinability-gap. Visualizations of other
datasets are similar and omitted.

\begin{figure}
\centering
\subfloat[DeepJoin]{
\begin{minipage}[t]{0.4\linewidth}
\centering
\includegraphics[width=1\linewidth]{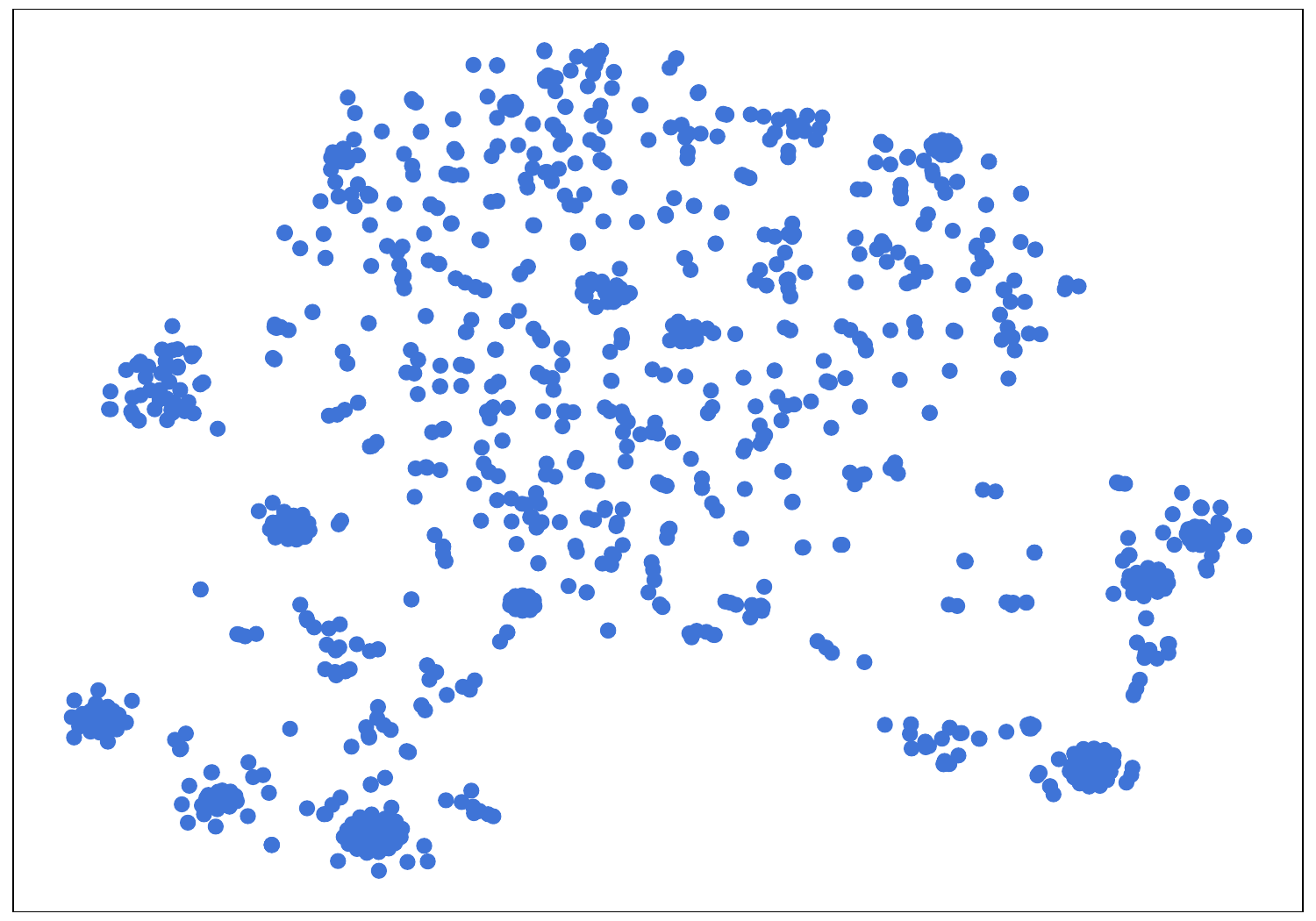}
    \vspace{-3mm}
\end{minipage}}
\hspace{7mm} 
\subfloat[\textsf{Snoopy}]{
\begin{minipage}[t]{0.4\linewidth}
\centering
\includegraphics[width=1\linewidth]{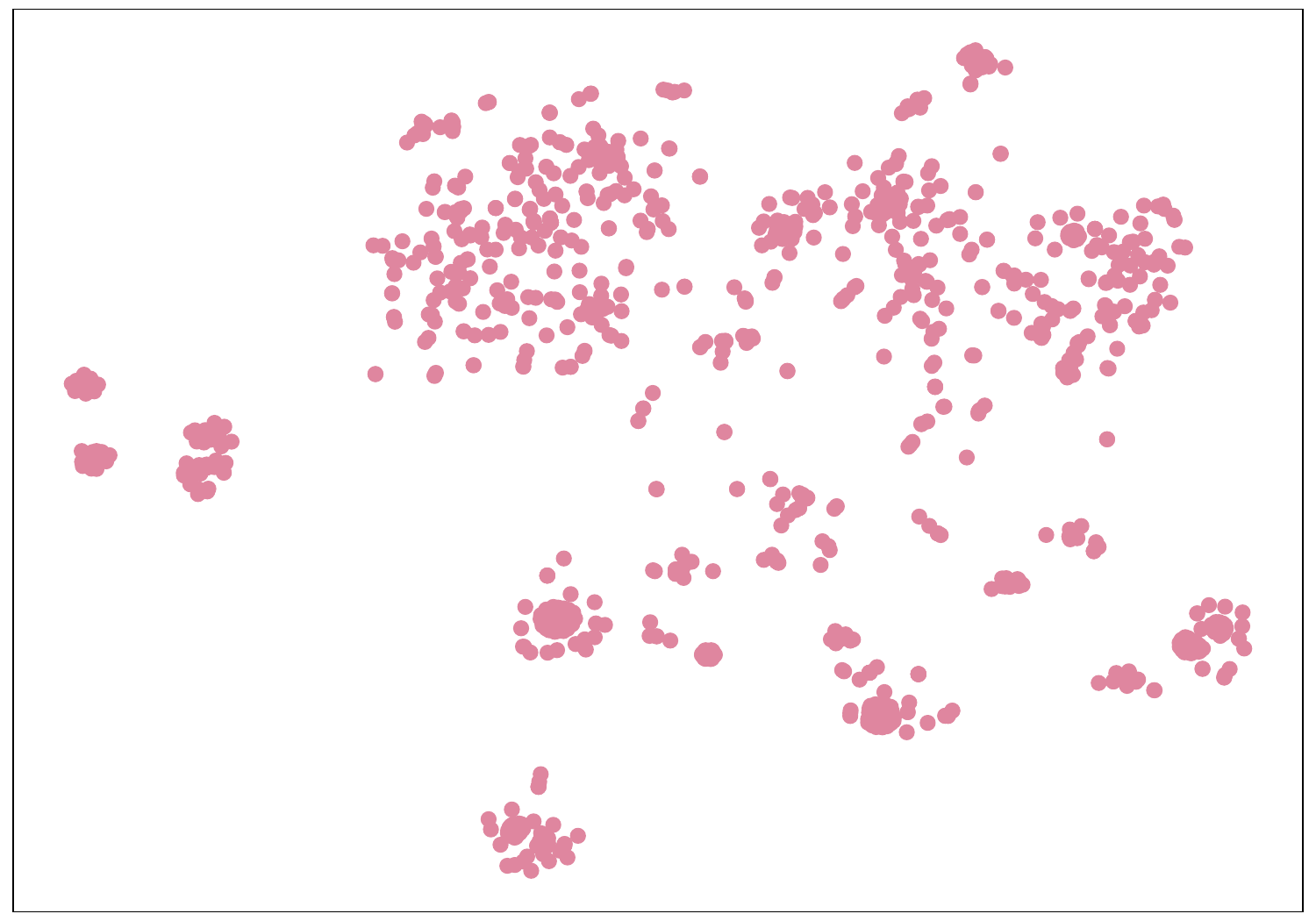}
    \vspace{-3mm}
\end{minipage}}
\caption{Visualization of column embeddings of Opendata learned by Deepjoin and our proposed \textsf{Snoopy}.
}
\label{fig:visual_embed}
\vspace{-3mm}
\end{figure}

\begin{figure}
\centering

\subfloat[DeepJoin]
{
\begin{minipage}[t]{0.45\linewidth}
\centering
\includegraphics[width=1\linewidth]{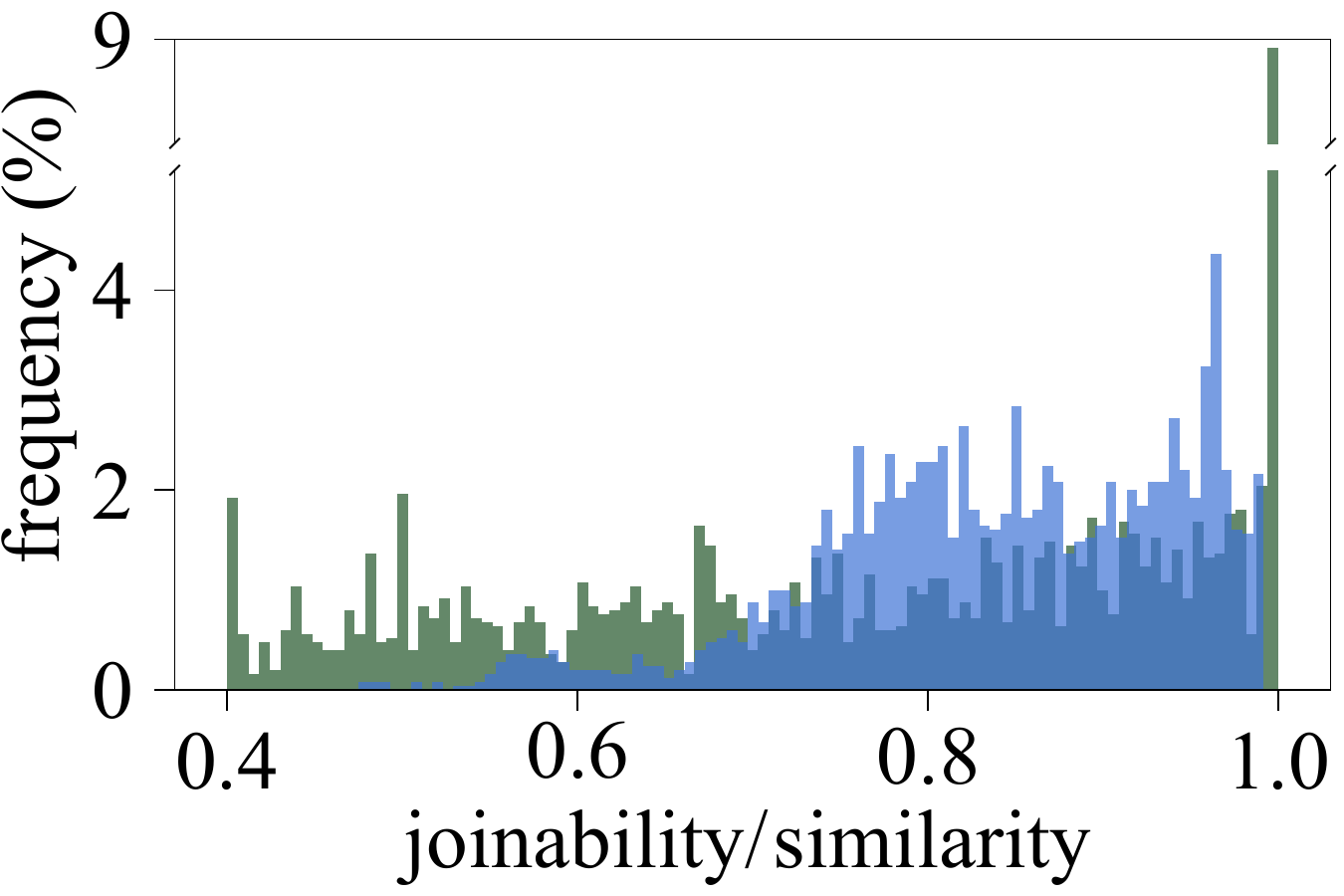}
 
    \vspace{-3mm}
\end{minipage}}
\hspace{3mm}
\subfloat[\textsf{Snoopy}]{
\begin{minipage}[t]{0.45\linewidth}
\centering
\includegraphics[width=1\linewidth]{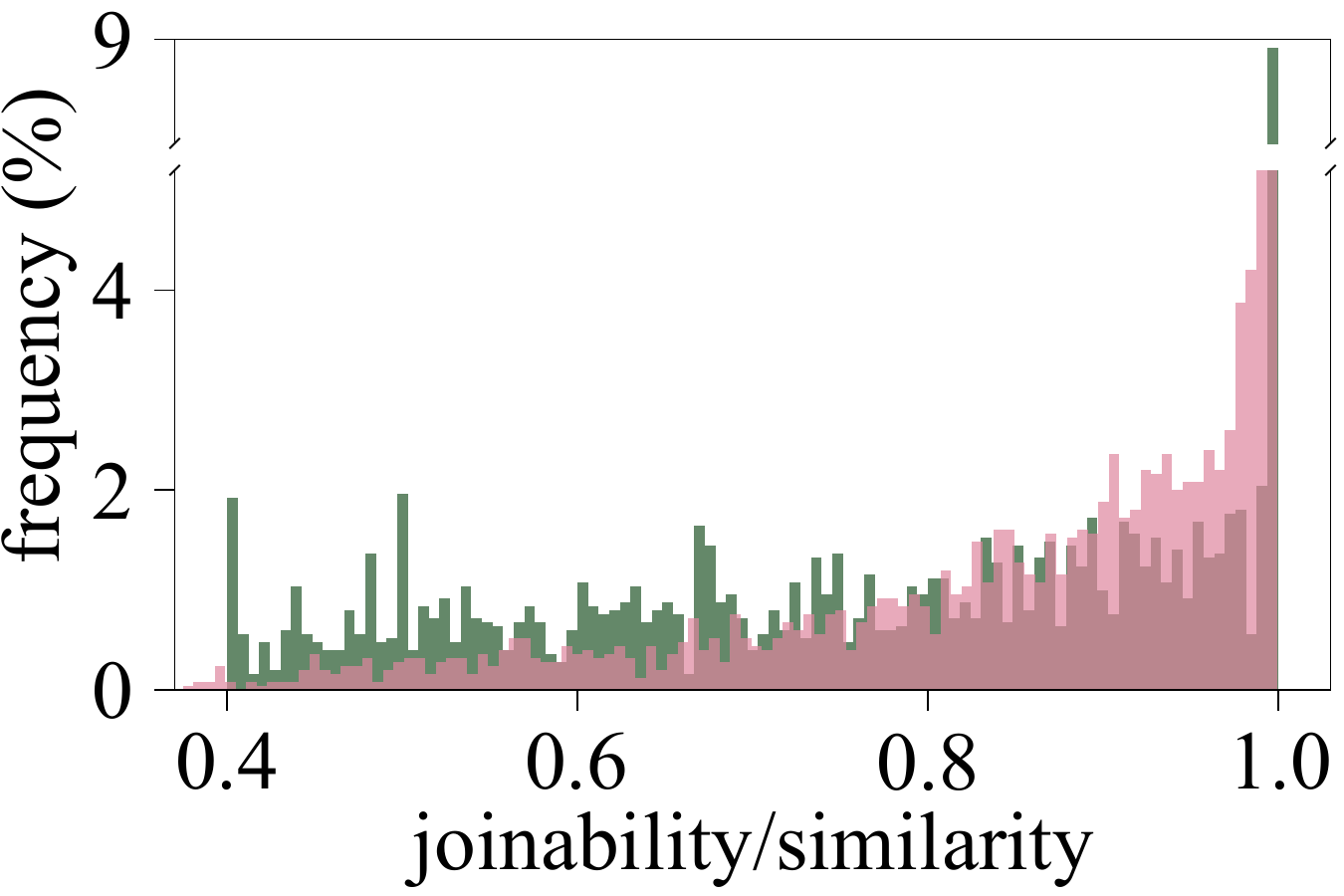}
 
    \vspace{-3mm}
\end{minipage}}
\caption{
Joinability distribution (in green) vs. column embedding similarity distribution obtained by Deepjoin and $\textsf{Snoopy}$.
}
\label{fig:visual_dist}
\vspace{-4mm}
\end{figure}


 
 

\begin{figure*}
\vspace{-3mm}
\centering
\begin{minipage}{\linewidth}
    \centering
       \includegraphics[width=0.5\textwidth]{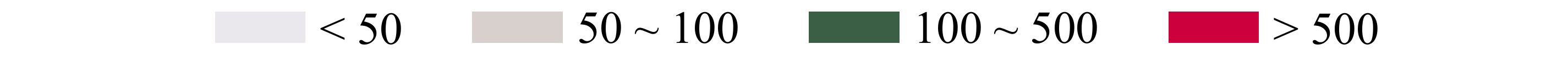}\\
       \vspace{-2mm}
  \end{minipage}
   
\subfloat[Recall@25 of different column sizes]{
\begin{minipage}[t]{0.48\linewidth}
\centering
\includegraphics[width=1\linewidth]{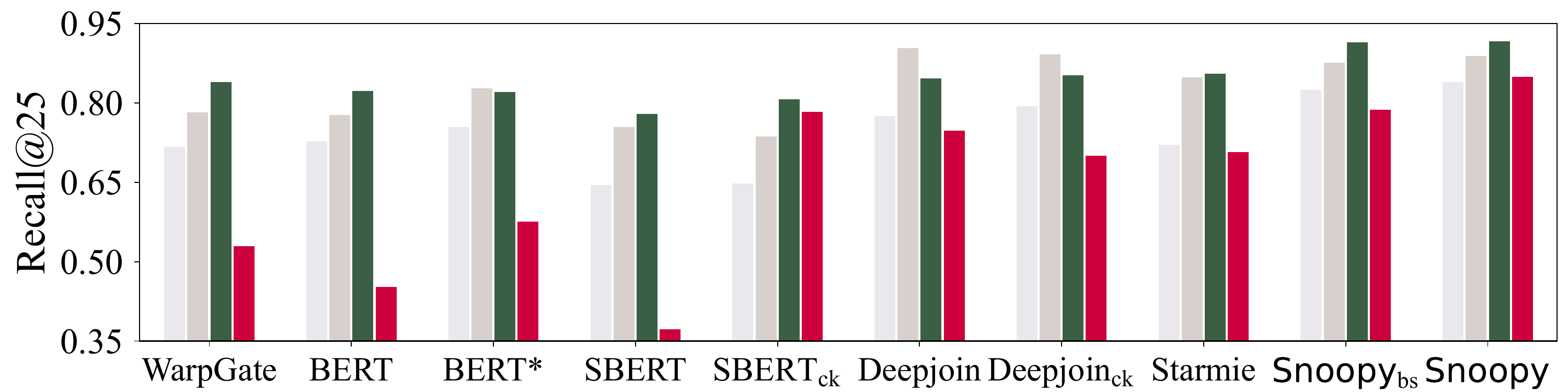}
    \vspace{-3mm}
\end{minipage}}
\subfloat[NDCG@25 of different column sizes]{
\begin{minipage}[t]{0.48\linewidth}
\centering
\includegraphics[width=1\linewidth]{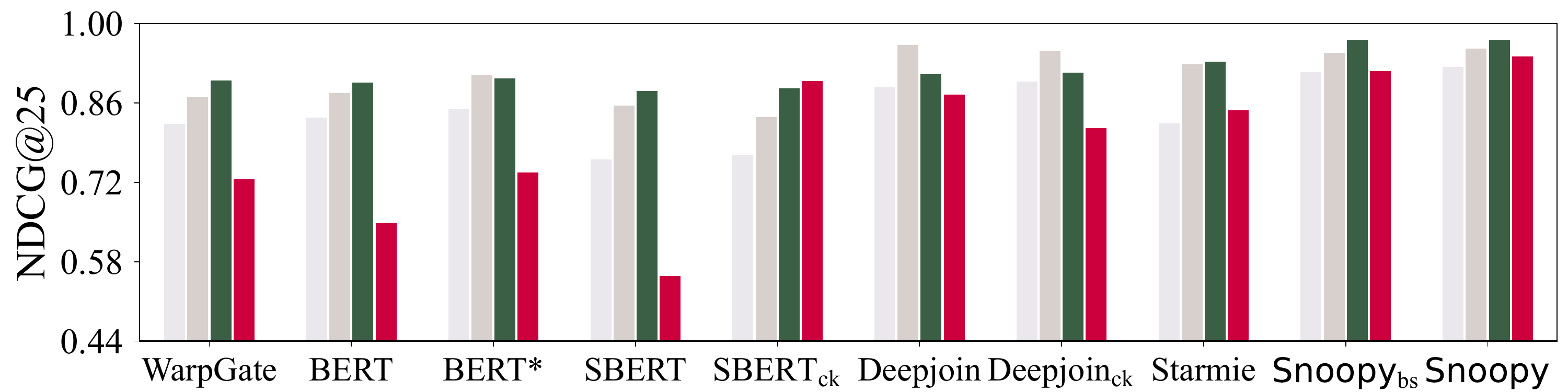}
    \vspace{-3mm}
\end{minipage}
}

\vspace{-2mm}
\caption{ Effectiveness evaluation in grouping columns of different sizes on the Opendata dataset.}  
\label{fig:impact of size}
\vspace{-4mm}
\end{figure*}

\noindent \textbf{Impact of column size.}
Then, we divide query columns of Opendata into four groups based on the column size ($<$50, 50-100, 100-500, and $>$500), and show the results on each group in Fig.~\ref{fig:impact of size}. 
This experiment is conducted on the Opendata dataset due to the limited number of long columns in Wikitable and short columns in WDC Small, respectively.
As observed, when the column size grows to more than 500, $\textsf{Snoopy}$ consistently maintains high performance, while most other baselines show a decline.
Notably, $\text{SBERT}_\text{ck}$ mitigates performance degradation of long columns on Opendata dataset. However, $\text{DeepJoin}_\text{ck}$ shows lower accuracy on long columns compared with $\text{DeepJoin}$. This discrepancy may stem from the inconsistency between fine-tuning (without chunking) and inference (with chunking). Unfortunately, the chunking strategy cannot be applied during training, limiting the overall performance.
Although DeepJoin employs a frequency-based sampling strategy to mitigates the negative impact of long columns, the performance is still not as stable as $\textsf{Snoopy}$. 


\noindent \textbf{Impact of permutation.} Finally, we explore the impact of cell permutation on effectiveness. Specifically, we randomly permutate the order of cells within each column in each dataset 50 times, and plot the distributions of Recall@25 in Fig.~\ref{fig:impact of order}. We omit the results of \textsf{Snoopy} because 
they are theoretically unaffected by permutations.
We can observe that the results of all the PTM-based methods are affected by the order of cells. In particular, SBERT exhibits the largest spread on all the datasets, attributed to its specialized sentence-level optimization, which considers sentence structure as an ordered sequence. Furthermore, the chunking strategy cannot mitigate the sensitivity, as each chunk is still processed as a sequential input by the PTMs, resulting in chunk embeddings being influenced by the cell order within each chunk.
In contrast, the  column representation of \textsf{Snoopy} is theoretically permutation-invariant, which is consistent with the definition of joinability that is agnostic to the cell orders.

Although $\textsf{Snoopy}$ primarily focuses on semantic join, 
it can also be applied to equi-join discovery by configuring the cell matching threshold $\tau = 0$ and removing data augmentation during training data generation. The effectiveness evaluation of equi-join discovery is detailed in Appendix B.

\begin{figure}
\centering

\subfloat{
\includegraphics[width=1\linewidth]{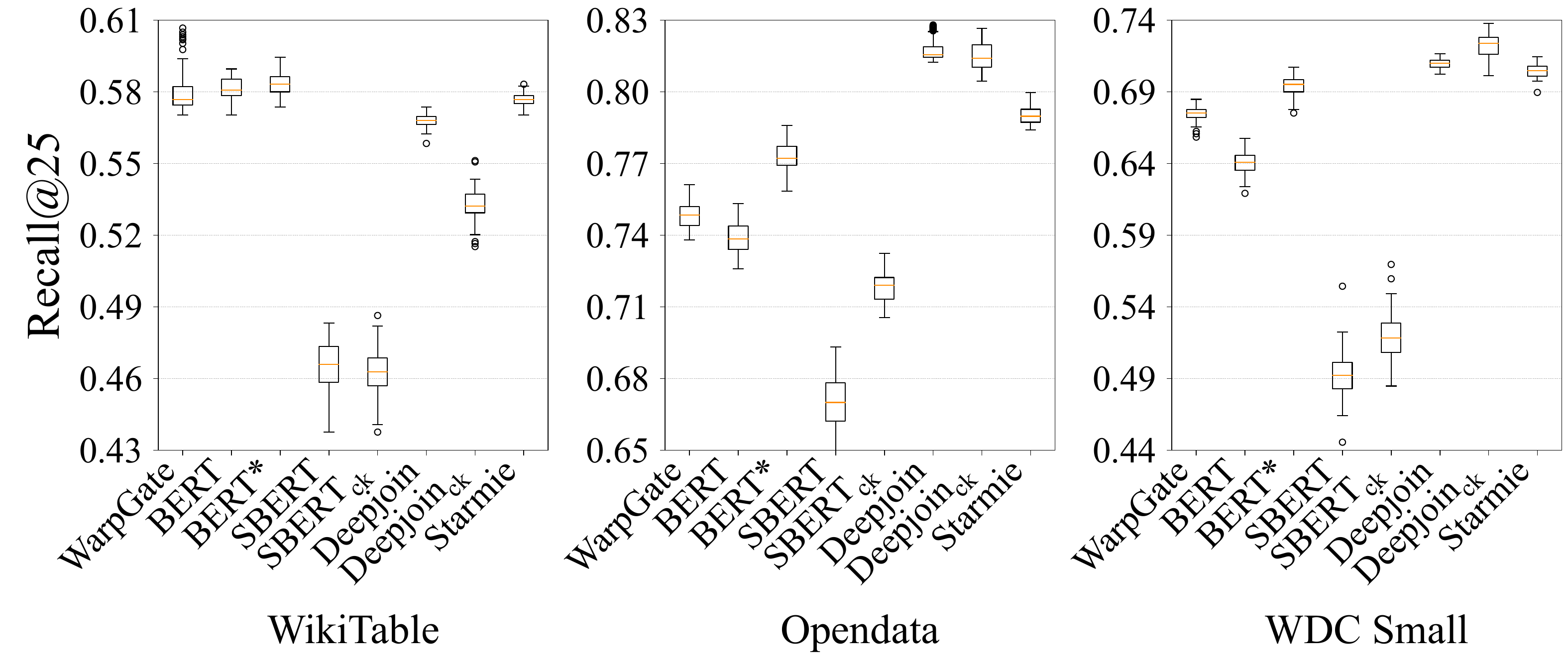}\hfill
}
\vspace{-4mm}
\caption{ Order sensitivity study by random permutation of cells. }
\label{fig:impact of order}
\vspace{-2mm}
\end{figure}


\subsection{Ablation Study}
\label{sec:exp_ablation} 

We conduct an ablation study of key components of \textsf{Snoopy}, with results shown in Table~\ref{tab:ablation}.

First, we replace the the AGM-based column projection (CP) with the widely-used pooling techniques, i.e., max-pooling  (MaxP),  min-pooling (MinP), and average-pooling (AvgP), to obtain the column embeddings. 
We can observe that the search accuracy dramatically drops.
This is because the widely used pooling methods result in much information loss when transforming the column matrices to column embeddings.  In contrast, our
AGM-based column projection well maintains the informative signals for joinability determination. 

Then, we remove the rank-aware contrastive learning (RCL) module, and extend the pivot selection methods in metric space~\cite{ZhuCGJ22} to the proxy column selection process:
(i) RAN selects columns from the repository $\mathcal{R}$ randomly as proxy columns; (ii) FFT~\cite{fft}  iteratively adds the new column  which is the most different from the current selected columns to the proxy column set; and (iii) PCA~\cite{pca} performs dimensionality reduction to select representative proxy columns based on FFT mechanism.
We can observe that without RCL, the Recall@25 drops at least 19.5\% on WikiTable, 12.9\% on Opendata, and 22.9\% on WDC Small. This demonstrates the effectiveness of RCL that is able to identify good proxy columns which yields promising effectiveness.

\begin{table} \small
\centering
\caption{Ablation study on three datasets.  Bold score indicates the performance under the default setting.}
\vspace{-2mm}
\renewcommand{\arraystretch}{1.2} 
\setlength{\tabcolsep}{0.8mm}{
\begin{tabular}{l|cc|cc|cc} 
\specialrule{.12em}{.06em}{.06em}
\multicolumn{1}{c|}{\multirow{2}{*}{\textbf{Methods}}} & \multicolumn{2}{c}{\textbf{WikiTable}} & \multicolumn{2}{c}{\textbf{Opendata}} & \multicolumn{2}{c}{\textbf{WDC Small}}  \\
\cline{2-7}
                         & \textbf{R@25}     & \textbf{N@25}               & \textbf{R@25}   & \textbf{N@25}                & \textbf{R@25}   & \textbf{N@25}                 \\ 
\hline
\multicolumn{1}{c|}{\textsf{Snoopy}}                & \textbf{0.7728} &  \textbf{0.9122}                   & \textbf{0.8716} &  \textbf{0.9458}                    & \textbf{0.8230} &    \textbf{0.9632}                   \\ 
\hline
  w/o CP + MaxP          &  0.6344      & 0.7791    &  0.7716     & 0.8725                 &    0.6984        & 0.9087                                            \\
  w/o CP + MinP                       &   0.6440       &  0.7784                  &     0.7268   &   0.8462                &   0.7216     &   0.9102                   \\
  w/o CP + AvgP                       &     0.7210     &   0.8780                 &   0.8300     &   0.9199                  &  0.7856      & 0.9500                     \\ 
\hline
  w/o RCL + RAN                       &  0.6184        &   0.7747                 &  0.7592      &  0.8678                   &  0.6080      &    0.8432                  \\
  w/o RCL + FFT                       &   0.6168       &  0.7677                  &  0.7188      &  0.8377                   &  0.6344      &   0.8652                   \\
  w/o RCL + PCA                       &   0.6224       &  0.7759                  &   0.7008     &  0.8267                   &   0.6064     &   0.8500                   \\
\specialrule{.12em}{.06em}{.06em}
\end{tabular}}
\label{tab:ablation}
\vspace{-3mm}
\end{table}

\subsection{Efficiency Evaluation}
\label{subsec:efficiency}

We report the runtime of $\textsf{Snoopy}$ and baseline methods. We omit the results of some methods because they demonstrate similar results to the specific methods already included. We also include PEXESO~\cite{Pexeso}, which is a cell-level exact solution.

We vary the number of columns in the WDC Large dataset from 100K to 1M, and report the average online processing time per query over 50 independent tests, as shown in Table~\ref{tab:online_efficiency}. 
Note that, for all column-level methods, we apply HNSW~\cite{HNSW} for ANN search.
We have the following observations:
(i) PEXESO and $\textsf{Snoopy}$ demonstrate high efficiency in online encoding. However, other PTM-based methods need a relatively longer online encoding time (10x longer for BERT*,  Starmie and DeepJoin, and 20x longer for WarpGate compared with \textsf{Snoopy}). This is because these transformer-based pre-trained models have complex architectures~\cite{FCS}, while \textsf{Snoopy} has a lightweight column projection mechanism. 
(ii) All the column-level methods significantly outperform the cell-level PEXESO in total time. Furthermore, \textsf{Snoopy} is 3.5x faster than other column-level methods on average, due to its shorter online encoding time.
(iii) CellSamp improves efficiency over PEXESO, but remains orders of magnitude slower than column-level methods due to the requirement of online pairwise cell similarity computations. 
(iv) Column-level methods demonstrate stable total time, even with an increase in dataset size. This is because the time of ANN search using HNSW is relatively stable, which is consistent with the prior study~\cite{Deepjoin}.
We also report the runtime of the offline stage in Appendix C.

\begin{table}[t] \small
\centering
\caption{Online processing time (ms) on WDC Large. Total time comprises query column encoding time and online search time.}
\vspace{-2mm}
\label{tab:online_efficiency}
\renewcommand{\arraystretch}{1.1} 
\begin{threeparttable}
\setlength{\tabcolsep}{0.42mm}{
\begin{tabular}{c|c|ccccc} 
\specialrule{.12em}{.06em}{.06em}
\multirow{2}{*}{\textbf{Methods}} & \multirow{2}{*}{\begin{tabular}[c]{@{}c@{}}query column \\ encoding  \end{tabular}} & \multicolumn{5}{c}{total online processing}                      \\ 
\cline{3-7}
                        &                                 & 100K      & 200K      & 300K  & 500K  & 1M     \\ 
\hline
PEXESO                  & 0.82                           & 311,761 & 655,011 &   \textsf{OOM}    &  \textsf{OOM}     &    \textsf{OOM}    \\
CellSamp    &  0.68         & 21,645    & 44,462 & 64,346 & 104,115 & 203,537 \\
\hline
WrapGate                & 20.56                           & 24.55     & 24.68     & 24.74 & 24.54 & 24.46  \\
BERT*                   & 8.89                            & 12.85     & 12.80     & 12.61 & 12.45 & 12.74  \\
Starmie                   & 9.22                          & 14.69     & 14.57     & 14.65 & 14.76 & 14.68     \\
DeepJoin                & 11.57                           & 15.56     & 15.54     & 15.56 & 15.54 & 15.57  \\
$\text{DeepJoin}_\text{ck}$             & 15.08     & 19.25      & 19.33      &19.22  &19.43   &19.82    \\
\hline
\textsf{Snoopy}               & 1.10                            & \textbf{5.10}      & \textbf{4.91}      & \textbf{5.00}  & \textbf{5.17}  & \textbf{4.97}   \\
\specialrule{.12em}{.06em}{.06em}
\end{tabular}}
\end{threeparttable}
\begin{tablenotes}\small
    \item  $^1$``\textsf{OOM}'' indicates out of memory under 128GB memory.
\end{tablenotes}
    \vspace{-2mm}
\end{table}

\begin{table}[t]
\small
\centering
\caption{ Ground truth shifts with different cell embedding functions.} 
\vspace{-2mm}
\label{tab:gt_shifts}
\renewcommand{\arraystretch}{1.1} 
\setlength{\tabcolsep}{1mm}{
\begin{tabular}{c|c|c|c} 
\specialrule{.12em}{.06em}{.06em}
\textbf{Cell embed. func.} & WikiTable & Opendata & WDC Small  \\ 
\hline
Word2vec          &       0.0341 &	0.0194 &	0.0212    \\
MPNet             &       0.0279 &  0.0085        &   0.0103         \\
\specialrule{.12em}{.06em}{.06em}
\end{tabular}}
\vspace{-6mm}
\end{table}

\begin{table*}[t]
\footnotesize
\centering
\vspace{-5mm}
\caption{ Performance  of \textsf{Snoopy} in the dynamic scenario. The best results are in bold.} 
\vspace{-2mm}
\label{tab:dynamic}
\renewcommand{\arraystretch}{1.3} 
\setlength{\tabcolsep}{0.75mm}{
\begin{tabular}{c|ccc|ccc|ccc|ccc|ccc|ccc}  
\specialrule{.12em}{.06em}{.06em}
\multirow{2}{*}{\textbf{Repository}}  & \multicolumn{6}{c|}{\textbf{Wikitable}}         &\multicolumn{6}{c|}{\textbf{Opendata}}         & \multicolumn{6}{c}{\textbf{WDC Small}}           \\ 
\cline{2-19}
 & \textbf{R@5} & \textbf{R@15} & \textbf{R@25} & \textbf{N@5} & \textbf{N@15} & \textbf{N@25} & \textbf{R@5} & \textbf{R@15} & \textbf{R@25} & \textbf{N@5} & \textbf{N@15} & \textbf{N@25} & \textbf{R@5} & \textbf{R@15} & \textbf{R@25} & \textbf{N@5} & \textbf{N@15} & \textbf{N@25} \\ 
\hline
$\mathcal{R}_0$      &  \textbf{0.5320}   &   \textbf{0.7347}   &  0.7576    &   \textbf{0.9270}  &  0.9194    &  0.8843   & \textbf{0.4860}   &  \textbf{0.8113}   &  0.7744    &  \textbf{0.9311}  & 0.9416   & 0.9338   & \textbf{0.6640} &  \textbf{0.7453}    &  0.7832    &   \textbf{0.9682}  &  0.9601     & 0.9524   \\


$\mathcal{R}_1$      &  0.5160   &   0.7090   &  0.7752    &   0.9255  &  \textbf{0.9253}    &  0.8909   &  0.4460  &  0.8000   &   0.8120   &  0.9310  &  \textbf{0.9448}  & 0.9420   & 0.6560 &  0.7360    &  0.8176    &  0.9590   &  \textbf{0.9605}     & 0.9568   \\


$\mathcal{R}_2$      &  0.4760   &   0.6813   &  0.7704    &   0.9200  &  0.9251    &  0.8841   &  0.4300  &  0.7567   &   0.8504   &  0.9348  & 0.9454   &  0.9475  & 0.6160 & 0.7053     & 0.8192     & 0.9668    &   0.9604    &  0.9618  \\


$\mathcal{R}_3$   &  0.5000   &   0.6600    &  \textbf{0.7728}   &   0.9187   &   0.9243   &   \textbf{0.9122}   & 0.3920     &   0.7180    &  \textbf{0.8716}    &   0.9300    &   0.9440    &   \textbf{0.9458}   &   0.5760   &   0.6827   &  \textbf{0.8230}    &  0.9613   &  0.9600    &    \textbf{0.9632}   \\

\specialrule{.12em}{.06em}{.06em}
\end{tabular}}
\vspace{-6mm}
\end{table*}

\subsection{Parameter Sensitivity}


\begin{figure}[t]
\centering
\subfloat[Varying $m$]{
    \includegraphics[width=0.45\linewidth]{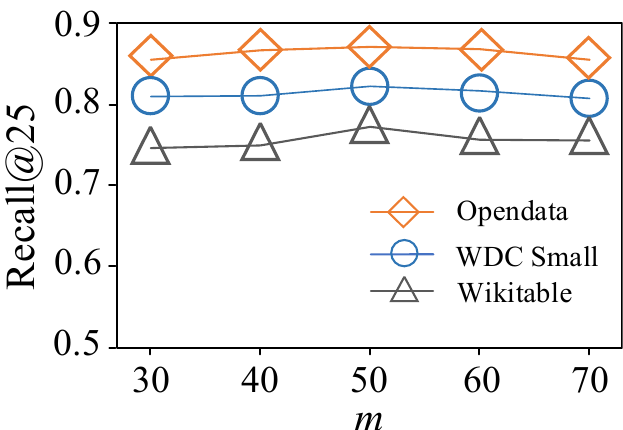}
}\hfill
\subfloat[Varying $l$]{
\includegraphics[width=0.45\linewidth]{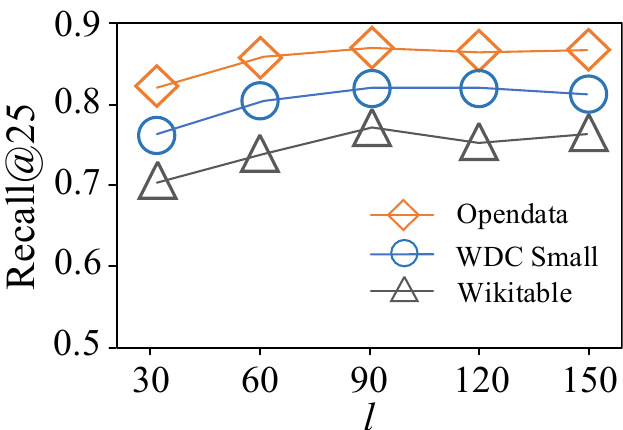}
}

\subfloat[Varying $\tau$]{
    \includegraphics[width=0.45\linewidth]{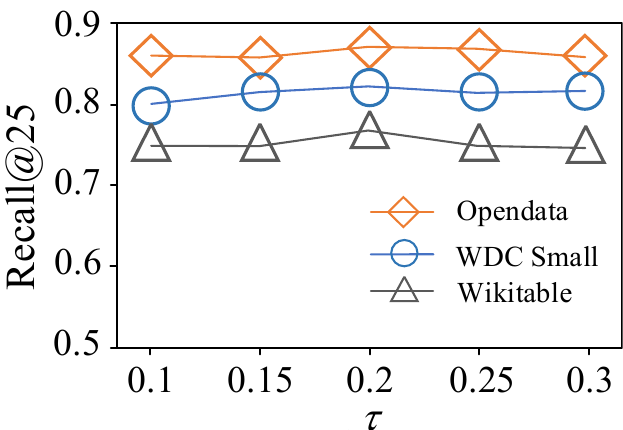}
}\hfill
\subfloat[Varying $s$]{
    \includegraphics[width=0.45\linewidth]{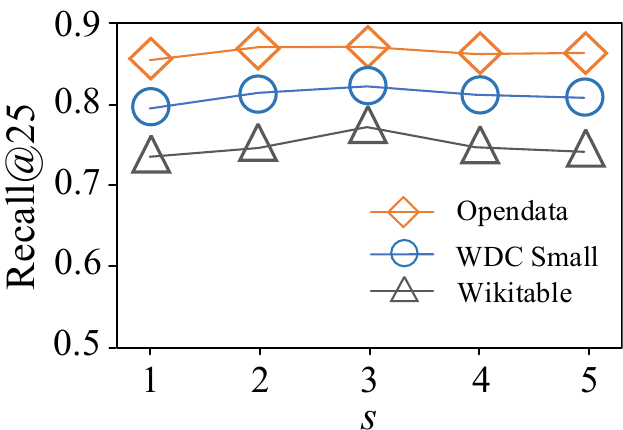}
}
\vspace{-2mm}
\caption{Sensitivity study of parameters of \textsf{Snoopy}.}
\label{fig:sense}
\vspace{-6mm}
\end{figure}

We study the influence of four important hyper-parameters on the performance of \textsf{Snoopy}.

First, we explore the impact of the number $m$ of elements in each proxy column. We vary the value of $m$ and show the result in Fig.~\ref{fig:sense}(a).
It is observed that the performance of \textsf{Snoopy} is not sensitive to the hyper-parameter $m$.

Next, we vary the number $l$ of used proxy columns, and show the results in Fig.~\ref{fig:sense}(b).
It is observed that as $l$ increases from 30 to 90, Recall@25  exhibits a noticeable improvement, which demonstrates that more information can be captured by increasing the number of proxy columns. When $l$ continues to increase, the recall no longer increases but tends to be stable. Thus, it is advisable to set $l$ to a relatively large value. For best performance across all datasets, we set $l$ to 90.

Then, we vary the threshold $\tau$ of cell matching from 0.1 to 0.3, and show the Recall@25 in Fig.~\ref{fig:sense}(c). It is observed that the performance is relatively stable under different thresholds $\tau$,  indicating that \textsf{Snoopy} is capable of accommodating different degrees of semantic join. Note when $\tau$ is set to 0, semantic-join degrades to equi-join.

Finally, we explore the impact of the length $s$ of the positive ranking list on the search results, as depicted in Fig.~\ref{fig:sense}(d). It is observed that as $s$ grows, the Recall@25 first increases and then gradually decreases. This is because treating the low-ranked columns in the ranking list as positive examples sometimes hurts the contrastive learning process, especially for the columns with high ranks. Hence, we set $s$ to 3 for the best performance on all the datasets.

\subsection{ Further Experiments} 
\label{subsec:further_exp}
We further (i) explore how ground truth in evaluation shifts when using different cell embedding functions; (ii) evaluate the effectiveness of \textsf{Snoopy} under the dynamic scenario; and (iii) explore some optimizations to overcome the size limits in existing PTM-based methods.

\noindent{ (i) \underline{\textit{Impact of cell embedding function}}}. 
We explore two cell embedding functions: Word2vec\footnote{https://huggingface.co/LoganKilpatrick/GoogleNews-vectors-negative300}, which is less powerful, and MPNet-based sentence-embedding model\footnote{https://huggingface.co/sentence-transformers/all-mpnet-base-v2}, which is more powerful than the default fastText.
Specifically, for each query column, we first obtain its top-25 ranked column list $L_f$ based on the semantic joinability using fastText. Then, we recompute joinabilities between the query column and those 25 columns using Word2vec and MPNet to generate new ranked lists $L_w$ and $L_m$, respectively. 
Finally, we measure ground truth shifts by quantifying the similarity between rankings $L_f$ and $L_w$ (resp. $L_m$) using $1-\rho =  \frac{6 \sum_{i=1}^s d_i^2}{s(s^2 - 1)} \in [0,1] $, where $\rho$ is the Spearman's rank correlation coefficient~\cite{Spearman}, $d_i$ represents the rank difference of column $C_i$ between $L_f$ and $L_w$ (resp. $L_m$), and $s$ is the list length.
The results are presented in Table~\ref{tab:gt_shifts}. We can observe that the shifts are small on all three datasets, demonstrating that while absolute embedding values may differ, the relative rankings derived by different cell embedding functions are similar. Notably, MPNet exhibits smaller shifts than Word2vec, indicating that the default fastText is closer to the MPNet than Word2vec.

\begin{table} \small
\centering
\caption{ Performance of E5-base-4k. $\Delta$ indicates the difference relative to the best-performing 512-token-limit PTMs.} 
\label{tab:long_ctx}
\vspace{-2mm}
\renewcommand{\arraystretch}{1.2} 
\setlength{\tabcolsep}{0.9mm}{
\begin{tabular}{c|cc|cc|cc} 
\specialrule{.12em}{.06em}{.06em}
\multicolumn{1}{c|}{\multirow{2}{*}{\textbf{Model}}} & \multicolumn{2}{c}{\textbf{WikiTable}} & \multicolumn{2}{c}{\textbf{Opendata}} & \multicolumn{2}{c}{\textbf{WDC Small}}  \\
\cline{2-7}
                         & \textbf{R@25}     & \textbf{N@25}               & \textbf{R@25}   & \textbf{N@25}                & \textbf{R@25}   & \textbf{N@25}                 \\ 
\hline
\multicolumn{1}{c|} {E5-base-4k}            &  0.5488	& 0.7306 & 0.8416	& 0.9165 & 0.6424 &	0.8646 \\
\hline
  $\Delta$          &  $\downarrow$0.0584      & $\downarrow$0.0285    &  $\uparrow$0.1196    & $\uparrow$0.0512                 &    $\downarrow$0.0616        & $\downarrow$0.0508                                            \\ 
\specialrule{.12em}{.06em}{.06em}
\end{tabular}}
\label{tab:ablation}
\vspace{-3mm}
\end{table}

\begin{table}[t] \small
\centering
\caption{ Online efficiency (ms) comparison of E5-base-4k and \textsf{Snoopy}.} 
\vspace{-2mm}
\label{tab:online_efficiency_E5}
\renewcommand{\arraystretch}{1.1} 
\begin{threeparttable}
\setlength{\tabcolsep}{0.9mm}{
\begin{tabular}{c|c|ccccc} 
\specialrule{.12em}{.06em}{.06em}
\multirow{2}{*}{\textbf{Methods}} & \multirow{2}{*}{\begin{tabular}[c]{@{}c@{}}query column \\ encoding  \end{tabular}} & \multicolumn{5}{c}{total online processing}                      \\ 
\cline{3-7}
                        &                                 & 100K      & 200K      & 300K  & 500K  & 1M     \\ 
\hline
E5-base-4k                  & 21.23                           & 25.12 & 25.49 &   25.38    &  25.67     &   25.71    \\
 
\hline
\textsf{Snoopy}               & 1.10                            & \textbf{5.10}      & \textbf{4.91}      & \textbf{5.00}  & \textbf{5.17}  & \textbf{4.97}   \\
\specialrule{.12em}{.06em}{.06em}
\end{tabular}}
\end{threeparttable}
\vspace{-6mm}
\end{table}

\begin{table}[t]
\small
\centering
\caption{ Performance of different cell sampling strategies applied to DeepJoin. The best results are in bold.} 
\label{tab:tfidf}
\vspace{-2mm}
\renewcommand{\arraystretch}{1.2} 
\setlength{\tabcolsep}{0.9mm}{
\begin{tabular}{c|cc|cc|cc} 
\specialrule{.12em}{.06em}{.06em}
\multicolumn{1}{c|}{\multirow{2}{*}{\textbf{Strategies}}} & \multicolumn{2}{c}{\textbf{WikiTable}} & \multicolumn{2}{c}{\textbf{Opendata}} & \multicolumn{2}{c}{\textbf{WDC Small}}  \\
\cline{2-7}
                         & \textbf{R@25}     & \textbf{N@25}               & \textbf{R@25}   & \textbf{N@25}                & \textbf{R@25}   & \textbf{N@25}                 \\ 

\hline
  {Frequency}          &  0.5600      & 0.6984    &  \textbf{0.8188}    & \textbf{0.9040} &      \textbf{0.7240} & \textbf{0.9170}                     \\ 
\hline
\multicolumn{1}{c|} {TF-IDF}       &  0.5608	 & 0.7061       &   0.8128 &	0.9016	    & 0.7104	& 0.9078 \\

  \hline
\multicolumn{1}{c|} {BM25}            &  \textbf{0.5611}	& \textbf{0.7066} 	& 0.8129 &	0.8955  &   0.7112	 & 0.9079	 \\
 
\specialrule{.12em}{.06em}{.06em}
\end{tabular}}
\label{tab:ablation}
\vspace{-3mm}
\end{table}

\noindent{ (ii){\underline{\textit{ Effectiveness in the dynamic scenario}}}.} To simulate a dynamic scenario where data are constantly added while fixing the number of proxy column matrices, we create an initial repository $\mathcal{R}_0$ by randomly sampling 70\% of the columns from the original dataset.
The remaining 30\%  are evenly divided into three batches $\mathcal{B}_1$, $\mathcal{B}_2$, and $\mathcal{B}_3$, which are incrementally added to the repository.
We denote the repository after adding the first $j$ batches as  $\mathcal{R}_j = \mathcal{R}_0 \cup \dots \cup \mathcal{B}_j (j \geq 1)$.
Table~\ref{tab:dynamic} presents the accuracy of \textsf{Snoopy} in the dynamic scenario. 
We observe some fluctuations of  R@$k$ in different sizes of data corpus. This is because the metric R@$k$ can be influenced by the cut-off errors due to the constraint of the fixed value of $k$, as mentioned in Sec.~\ref{sec:exp_effectiveness}. However, \textsf{Snoopy} demonstrates relatively stable and high N@$k$ performance as new columns are added, highlighting its effectiveness in dynamic scenarios. We also observe that as more data are added, N@$5$ decreases slightly, while N@$25$ increases. 
This arises from the sparsity of joinable columns in the data repository. Adding more data enhances the likelihood of identifying additional joinable columns, improving N@$25$. However, this also increases the difficulty of identifying the most joinable columns, resulting in a slight decrease in N@$5$.


\noindent{(iii) {\underline{\textit{PTM-based methods optimized for size limits}}}.} First, we apply the long-context model E5-Base-4k~\cite{E5-Base-4k}, which supports 4k tokens, to encode columns. The results compared with the best-performing 512-token-limit models are shown in Table~\ref{tab:long_ctx}.
As observed, the accuracy increases on Opendata but decreases on the other two datasets. This is because,  while the long-context model overcomes the size limit, the semantics-joinability gap persists. As more cells are added, non-matching cells may dominate the column's embedding, resulting in reduced accuracy. Furthermore, the long-context model's online encoding is time-consuming, requiring on average 20$\times$ more query column encoding time compared with \textsf{Snoopy}, as shown in Table~\ref{tab:online_efficiency_E5}.
Since the sentence-transformers\footnote{https://huggingface.co/sentence-transformers} does not yet support long-context models, we test the accuracy using the pre-trained E5-Base-4k, expecting similar trends if applied to DeepJoin.
Second, we explore TF-IDF and BM25 to sample cells within each column to ensure the 512-token limit for DeepJoin. The results compared with the original DeepJoin (frequency-based) are shown in Table~\ref{tab:tfidf}. As observed, the impact of different strategies on accuracy is quite limited. This is because these sampling strategies are designed for text similarity and information retrieval tasks, which do not align well with the definition of joinability.

Based on these observations, we compare the pros and cons of \textsf{Snoopy} versus DeepJoin with a long-context encoder. Both \textsf{Snoopy} and long-context-model-powered DeepJoin can overcome size limitations, however, \textsf{Snoopy} has two advantages: (i) it effectively bridges the semantics-joinability gap through proxy column matrices, whereas the long-context model struggles with this gap and may even exacerbate it; and (ii) \textsf{Snoopy} is much more efficient than long-context-model during online search. The primary limitation of \textsf{Snoopy} is its reliance on the cell embedding function and join definition, whereas DeepJoin offers greater flexibility by adapting to various PTMs and join definitions.

\section{Related Work}
\label{sec:relatedwork}
\subsection{Dataset Discovery} 
Dataset discovery has been widely studied in the data management community~\cite{paton2023dataset, TabelDiscovery}, with table search  as the primary application. 
The main sub-tasks of table search are joinable table search and table union search.

\textbf{Joinable table search.}
To support joinable table search, most studies~\cite{Aurum, DatasetDiscovery,LSH,JOSIE,CrossDataDis,Correlation,MATE} focus on equi-join and utilize syntactic similarity measures to determine joinanility between columns. 
To take semantics into consideration, PEXESO~\cite{Pexeso} proposes a semantic joinability measure, and designs a cell-level exact algorithm under this measure using word embeddings. To enhance efficiency, the following column-level solutions, such as DeepJoin~\cite{Deepjoin} and WarpGate~\cite{WarpGate}, perform coarser computation at column-level to approximate the results of cell-level solutions. However, the effectiveness is poor due to the suboptimal column embeddings. In contrast, our $\textsf{Snoopy}$ is an effective column-level framework powered by the proxy-column-based column embeddings.
Recent works~\cite{koios,SilkMoth} on semantic overlap set search are related to join discovery, but adopt a different semantic overlap (join) measure from the previous studies~\cite{Deepjoin,Pexeso} and ours. 
OmniMatch~\cite{omnimatch}, a concurrent work with ours, detects both equi-joins and fuzzy-joins by combining multiple similarity measures. However, it views join discovery as an offline procedure~\cite{omnimatch}, unlike online procedures where high efficiency is a crucial demand.

\textbf{Table union search.} The goal of table union search is to find tables that can be unioned with the query table to extend it with tuples. TUS~\cite{TUS} defines table union search based on attribute unionability, and formalizes three probabilistic models to describe how unionable attributes are
generated from different domains.
D3L~\cite{DatasetDiscovery} further adds in measures that include formatting similarity and attribute names. SANTOS~\cite{santos}
considers the relationships between columns and uses a knowledge base to identify  unionable tables. Starmie~\cite{starmine} extends the notion of capturing binary relationships to use the context of the table to determine union-ability.

\subsection{Table Representation}
Many researchers are exploring how to represent tabular data (i.e., structured data) with neural models~\cite{tableembed, badaro2023transformers}. Due to the huge success in natural language processing (NLP), pre-trained language models (e.g., BERT~\cite{bert}, SBERT~\cite{sentencebert}, E5-base~\cite{E5-Base-4k}, etc.) have been widely applied to represent different levels of tabular data, including entity matching (row-level)~\cite{camper,ditto}, column type annotation (column-level)~\cite{doduo,Watchog}, etc. To model the row-and-column structure as well as integrate the heterogeneous information from different components of tables, transformer-based table embedding models have been proposed, such as TURL~\cite{turl}, TaBERT~\cite{tabert}, TAPAS~\cite{tapas}, etc. These models are based on Transformer architecture, and thus, enforce a length limit to token sequences (e.g. 512) due to the high computational complexity of the self-attention mechanism~\cite{attention}. In contrast, our proposed column representation is size unlimited, which can well handle the long columns in the real table repositories.

\subsection{Contrastive Learning}
Contrastive learning~\cite{Moco} (CL) is a discriminative approach that aims to pull similar samples closer and push apart dissimilar ones in the embedding space, and has achieved huge success in diverse domains.
In data discovery and preparation, CL is an effective method for learning high-quality data representations. Pylon~\cite{Pylon} and Starmie~\cite{starmine} leverage CL to learn column representations for table discovery. Ember~\cite{Ember}  enables a general keyless join operator by constructing an index populated with task-specific embeddings via CL. Sudowoodo~\cite{Sudowoodo} applies CL to learn entity, column, and cell representations to address multiple tasks in data preparation. Instead of directly learning the data items, in this paper, we leverage CL to learn proxy column matrices, which are then used to derive column representations efficiently. In contrast to the standard CL that requires a strict binary separation of the training pairs into similar and dissimilar samples, RINCE~\cite{2022ranking} first proposes a new mechanism that preserves a ranked ordering of positive samples. Inspired by that, we incorporate rank awareness into the pivot column matrix learning process, focusing on designing a data generation method to synthesize ranked joinable columns to enable rank-aware CL.

\section{Conclusions}
\label{sec:conlusion}
This paper proposes \textsf{Snoopy}, an effective and efficient semantic join discovery framework powered by proxy columns.
We devise an approximate-graph-matching-based column projection function to capture column-to-proxy-column relationships, ensuring size-unlimited and permutation-invariant column representations.
To acquire good proxy columns, we present a rank-aware contrastive learning paradigm to learn proxy column matrices for embedding pre-computing and online query encoding.
Extensive experiments on four real-world table repositories demonstrate the superiority of \textsf{Snoopy} in both effectiveness and efficiency. In the future, we plan to consider the downstream data analysis and study the task-oriented join discovery.
\appendices

\section{Time complexity analysis}
\label{appendix:A}
We denote the cardinality of proxy column set  as $l$ and the cardinality per proxy column as $m$. During the offline stage, the complexity of pre-computing all the column embeddings is $\mathcal{O}(ml|\overline{C}||\mathcal{R}|) = \mathcal{O}(|\overline{C}||\mathcal{R}|)$, where $|\mathcal{R}|$ denotes the number of columns in the repository, and $|\overline{C}|$ denotes the average size of columns in the repository; and the time complexity of index construction using HNSW is $\mathcal{O}(|\mathcal{R}| \operatorname{log}|\mathcal{R}|)$. During the online stage, the complexity of computing the query column embedding is $\mathcal{O}(ml|C_Q|) = \mathcal{O}(|C_Q|)$, where $|C_Q|$ is the size of the query column $C_Q$; and the time complexity of ANN search is  $\mathcal{O}(\operatorname{log}|\mathcal{R}|)$. 
Consequently, the overall time complexity of online processing is $\mathcal{O}(|C_Q| + \operatorname{log}|\mathcal{R}|)$.

\section{Effectiveness of equi-join discovery}
\label{appendix:B}
To demonstrate the effectiveness of \textsf{Snoopy} in equi-join search, we compare its R@25 with Deepjoin, which can also deal with approximate equi-join discovery.
We omit the exact algorithms such as JOSIE~\cite{JOSIE} and MATE~\cite{MATE}, as their accuracy is 1. The approximate algorithm like LSH Ensemble~\cite{LSH} is also excluded, as it has been validated that Deepjoin outperforms LSH Ensemble in equi-join search~\cite{Deepjoin}. The results are shown in Table~\ref{tab:equi-join}. We can see that the accuracy of $\textsf{Scorpion}$  also outperforms Deepjoin.

\begin{table}[h]
\small
\centering
\caption{Accuracy of equi-join search. The best are in bold.}
\vspace{-3mm}
\renewcommand{\arraystretch}{1.2} 
\setlength{\tabcolsep}{1.4mm}{
\begin{tabular}{l|cc|cc|cc} 
\specialrule{.12em}{.06em}{.06em}
\multicolumn{1}{c|}{\multirow{2}{*}{\textbf{Methods}}} & \multicolumn{2}{c}{\textbf{WikiTable}} & \multicolumn{2}{c}{\textbf{Opendata}} & \multicolumn{2}{c}{\textbf{WDC Small}}  \\
\cline{2-7}
                         & \textbf{R@25}     & \textbf{N@25}               & \textbf{R@25}   & \textbf{N@25}                & \textbf{R@25}   & \textbf{N@25}                 \\ 
\hline
\multicolumn{1}{c|} {Deepjoin}  & 0.4984 &  0.6498   & 0.8168 &  0.8992    &   0.6960 &  0.9026                   \\ 
  {$\textsf{Snoopy}_\text{bs}$}          &   0.7288     &   0.8985   &   0.8380    &   0.9297               &      0.7840     &       0.9474                                       \\
  {\textsf{Snoopy}}                       &   \textbf{0.7416}      &    \textbf{0.9077}               &   \textbf{0.8576}   &    \textbf{0.9449}               &    \textbf{0.8096}    &     \textbf{0.9632}                 \\
\specialrule{.12em}{.06em}{.06em}
\end{tabular}}
\label{tab:equi-join}
\vspace{-3mm}
\end{table}

\newpage
\section{Offline Efficiency Evaluation}
\label{appendix:C}
The runtime (in minutes) of the offline process is shown in Table~\ref{tab:offline_efficiency}. We can observe that $\textsf{Snoopy}$ exhibits the shortest per-epoch training time due to its lightweight AGM-based column mapping. PEXESO requires the least encoding time, as it simply invokes fastText to encode each cell within the query column.
The encoding time of Snoopy is also short
due to the lightweight AGM-based column projection.  Since all column-level methods employ HNSW for indexing, their indexing times are comparable. The indexing time of PEXESO is long 
due to the necessity of indexing the mapped vectors of all cells in hierarchical grids~\cite{Pexeso}.

\begin{table}[h]
\small
\centering
\caption{Offline processing time (min). The training is evaluated on WDC Small while encoding and indexing are on WDC Large.}
\vspace{-2mm}
\renewcommand{\arraystretch}{1.1} 
\begin{threeparttable}
\setlength{\tabcolsep}{3.3mm}{
\begin{tabular}{c|c|c|c} 
\specialrule{.12em}{.06em}{.06em}
\textbf{Methods}  & training/epoch & encoding & indexing  \\ 
\hline
PEXESO   &      --             &   \textbf{8.75}        &       298.44    \\
WarpGate &      --             &    115.98        &    0.29        \\
BERT*    &       36.54      &    110.01        &     0.32        \\
Starmie &         41.60   &   128.54       &    0.35          \\
Deepjoin &         70.46   &   120.52        &    0.39          \\ 
\hline
\textsf{Snoopy} &   \textbf{2.03} &      14.81              &     \textbf{0.28}           \\
\specialrule{.12em}{.06em}{.06em}
\end{tabular}}
\end{threeparttable}
    \vspace{-2mm}
    \label{tab:offline_efficiency}
\end{table}



 
%

\bibliographystyle{IEEEtran}
\bibliography{IEEEabrv,sample-base}

\vfill

\end{document}